\documentclass[journal]{IEEEtran}

\usepackage{booktabs}
\usepackage{amssymb}
\usepackage{amsmath}
\usepackage{bm}
\usepackage{nccmath}
\usepackage{enumitem}
\usepackage{algorithm}
\usepackage{algpseudocode}
\usepackage{amssymb}
\usepackage{graphicx}
\usepackage{caption}
\usepackage{xcolor}

\newtheorem{definition}{Definition}[section]
\newtheorem{lemma}{Lemma}[section]
\newtheorem{assumption}{Assumption}[section]
\newenvironment{proof}{\begin{IEEEproof}}{\end{IEEEproof}}

\newtheorem{theorem}{Theorem}[section]
\newtheorem{remark}{Remark}[section]

\ifCLASSINFOpdf

\else

\fi

\begin{document}

\title{Deception Equilibrium Analysis for Three-Party Stackelberg Game with Insider}

\author{Xiaoyu~Xin,
        Gehui~Xu,
        Yiguang~Hong
}



\author{Xiaoyu~Xin, Gehui~Xu, and Yiguang~Hong%
\thanks{This work was supported by the National Key
Research and Development Program of China under Grant 2022YFA1004700,
the National Natural Science Foundation of China under Grant 62573319, and Shanghai Municipal Science
and Technology Major Project under Grant 2021SHZDZX0100.}%
\thanks{X. Xin and Y. Hong are with the Shanghai
Research Institute for Intelligent Autonomous Systems, Tongji University, Shanghai, 201210, China, and Yiguang Hong is also
with the Department of Control Science and Engineering, Tongji University, Shanghai, 201804, China (e-mail: 2151471@tongji.edu.cn; yghong@tongji.edu.cn).}%
\thanks{G. Xu is with the Department of Electrical and Electronic Engineering, Imperial College London, London SW7 2AZ, UK (e-mail: g.xu@imperial.ac.uk).}%
}

\maketitle

\begin{abstract}
    This paper investigates strategic interactions within a three-party deception security game involving a defender, an insider, and external attackers. We propose a robust deception mechanism where the leader manipulates game parameters perceived by followers to enhance defense performance when followers
operate under misperceived and uncertain observation. Specifically, we propose a unified three-party leader–follower game framework and introduce the concepts of Deception Stackelberg equilibria (DSE) and Hyper Nash equilibria (HNE), which generalize classical two-player Stackelberg and deception games.
We develop necessary and sufficient conditions for the consistency between DSE and HNE, ensuring that the defender's utility remains invariant when the hierarchical structure degenerates into a simultaneous-move scenario. Moreover,  we propose a scalable hypergradient-based algorithm with established convergence guarantees for seeking DSE, efficiently addressing the computational challenges posed by non-smooth and set-valued best-response mappings. Finally, we apply theoretical analysis to practical scenarios in secure wireless communication and defense against insider-assisted false data injection attacks. 
\end{abstract}

\begin{IEEEkeywords}
    Three-party game; Deception Stackelberg equilibrium; Hyper Nash equilibrium; Hypergradient-based algorithm
\end{IEEEkeywords}

\IEEEpeerreviewmaketitle

\section{Introduction}

\IEEEPARstart{S}ECURITY games describe scenarios in which a protected system defends against malicious attacks, and have been widely applied in cybersecurity problems, such as wireless communication security \cite{8454822,7801047}, defense against insider-
assisted false data injection (IA-FDI) attacks in microgrids \cite{9863873,gonen2020false}, and adversarial machine learning \cite{10273619}.
Beyond conventional two-party attacker–defender models, the threat posed by insiders as third parties has emerged as a critical yet often overlooked aspect of cybersecurity.
According to a recent global report \cite{PonemonDTEX2025}, the average annual cost associated with insider-related incidents increased by nearly 50\% between 2019 and 2025.
An insider, with privileged access to system resources, defense strategies, and sensitive information, may deliberately or inadvertently expose critical internal information to external attackers~\cite{cappelli2012cert}.
As a result, three-party security games involving a defender, an insider, and one or multiple attackers have become an emerging research topic.
The corresponding classical decision-making paradigm for such interactions is the leader–follower model, in which the defender as the leader dominates the decision process by anticipating the follower’s reaction, while the follower selects its best-response (BR) strategy after observing the leader’s action.
The corresponding well-known equilibrium is the Stackelberg equilibrium (SE).  

Typically, Stackelberg games in the literature assume that each player’s perceived and observed information accurately reflects the underlying environment~\cite{10520318}.
However, misinformation from deception involves the active manipulation of followers’ observations through actions such as belief manipulation, information hiding, or camouflage, and is prevalent in many scenarios~\cite{10.1145/3214305}.
For example, in secure wireless communications~\cite{FANG2017153}, a source node may forge channel state information to influence an eavesdropper’s jamming strategy, thereby improving the secure transmission rate.
Similarly, in microgrids~\cite{liu2021defense}, a defender may deliberately disclose signals indicating stricter monitoring and scrutiny to induce insiders to cooperate with internal security mechanisms, rather than underestimating the risk of betrayal and leaking sensitive information to false data injection attackers for personal gain.
To model strategic interactions in deceptive environments, the Deception Stackelberg equilibrium (DSE) has been introduced \cite{Cheng_2022, nguyen2019imitative}, in which the leader manipulates followers’ perceptions while optimizing its own utility.
When followers’ BR mappings are set-valued, two classical tie-breaking assumptions arise, leading to the Weak and Strong Deception Stackelberg equilibria (WDSE and SDSE) \cite{loridan1996weak,guo2018inducibilitystackelbergequilibriumsecurity},
characterizing the lower and upper bounds of the leader’s achievable utility, respectively.

Although a leader may exploit inherent information asymmetry to mislead followers, in practice, followers may lack the ability or incentive to adopt BR strategies due to limited observation capabilities~\cite{chen2019interdependent}, environmental disturbances~\cite{nguyen2009security}, or intentional information concealment~\cite{8485952}.
For example, in wireless interference scenarios, an interferer may suffer from observation errors caused by uncertainty in time-varying channel states \cite{7076591}.
In cyber-physical power systems, strict confidentiality of security configurations and operational compartmentalization may prevent an insider from observing the defender’s specific strategy \cite{9863873}. Consequently, the leader cannot ascertain whether followers will adhere to the leader–follower paradigm and thus cannot guarantee the preservation of its dominance under deception.
Hypergame theory provides a framework for analyzing strategic interactions under misinformation and heterogeneous player cognitions in non-dominant settings.
Its central idea is to decompose a complex interaction into multiple subjective games, each reflecting a player’s own perception of the strategic environment \cite{kovach2015hypergame}.
The corresponding solution concept is the Hyper Nash equilibrium (HNE) \cite{ne,Sasaki_2008}, in which each player adopts a BR strategy within its own subjective game. By shaping followers’ perceptions through deceptive signaling, the leader can align the DSE with the HNE, thereby preserving its utility despite the loss of hierarchical dominance.

Beyond the robustness of deception strategies, the computation of DSE also needs investigation, due to the non-smooth and potentially set-valued nature of followers’ BR mappings.
While several recent works have studied hierarchical game problems, existing three-level game models largely lack effective algorithms and related convergence for computing optimal deception strategies~\cite{NEURIPS2021_3de568f8}.
Moreover, most available hierarchical optimization methods are tailored to traditional bilevel games and rely on single-valued BR assumptions, thereby overlooking the tie-breaking issues that naturally arise in practical deception scenarios~\cite{Dempe04032019}.

Therefore, the motivation of this paper is to design optimal and utility-robust deception strategies for the leader in a three-party security game under followers' perception bias and observation uncertainty. 


\subsection{Contributions}
\begin{enumerate}
  \item We formulate a deception three-party Stackelberg game that incorporates active deception into hierarchical decision-making. The proposed model provides a unified formulation that includes the original three-player setting without misinformation and the two-level leader–follower misinformation game as instances. We establish the existence conditions for an SDSE and a WDSE.


\item We establish a necessary and sufficient condition under which each WDSE coincides with an HNE. An analogous condition holds for SDSE. This result guarantees robustness of the leader’s utility under follower behavioral uncertainty that may eliminate the leader’s dominant position.

\item 

We propose a scalable hypergradient-based algorithm and establish its convergence to the WDSE and SDSE. Moreover, when the BR cannot be exactly obtained, or WDSE may not exist, the proposed algorithm is guaranteed to converge to an $\epsilon$-WDSE. Numerical case studies in secure wireless communication and insider-assisted false data injection defense verify the theoretical findings and demonstrate the effectiveness of the algorithm.

\end{enumerate}

\subsection{Related work}
Hierarchical decision-making models have been extensively studied to characterize interactions in complex systems, ranging from cybersecurity to network management. While early works focused on two-player interactions, recent research has shifted towards three-party hierarchical structures to capture cascading strategic effects. For instance, in secure wireless networks, \cite{8281119} modeled a macro base station (MBS) managing interfering small base stations (SBSs) to thwart eavesdroppers, creating a tri-level resource allocation chain. Similarly, \cite{8927954} analyzed a tiered pricing game where a top-layer source prices energy for a mid-level interferer based on bottom-layer constraints. These studies establish the foundation of single-leader-multi-follower frameworks. However, most existing works assume perfect information flow between layers, neglecting the strategic implications of observation failures or manipulated signaling in adversarial contexts. 

Misinformed games are not limited to passive observation errors but also encompass strategic deception, when players deliberately manipulate others’ beliefs or perceptions.
Such interactions are commonly studied using Bayesian game models or the Hypergame framework.
Bayesian games rely on the common prior assumption,   differing only in their private information \cite{SHI202249}.
However, in strategic deception scenarios, the deceiver’s objective is often to induce a fundamental misconception of the game itself, such as the opponent’s perceived strategy sets or utility functions \cite{Cheng_2022,ZHAO2024121619}.
These forms of cognitive manipulation violate the common prior assumption underlying Bayesian games.
The Hypergame framework explicitly allows players to hold subjective and potentially inconsistent representations of the game, thereby providing a more natural and direct modeling tool for strategic deception driven by cognitive misalignment.


The choice of equilibrium is always crucial in strategic decision-making.
NE and SE are two classical solution concepts for simultaneous and sequential decision-making schemes, respectively.
The relationship between NE and SE has been extensively studied in differential game settings \cite{10520318,10.5555,sengupta2019general}.
However, under strategic deception or perception bias, players may optimize against misperceived objectives or strategy sets, fundamentally altering the equilibrium structure.
This challenge is further compounded in multi-agent settings with an unidentified third-party insider, leading to hierarchical interactions beyond the traditional two-player paradigm.
In such three-party hierarchical hypergames, the relationship between DSE and HNE remains largely unexplored.



Solving three-party game problems under strategic deception is computationally challenging, largely due to the non-smooth and potentially set-valued nature of best-response (BR) mappings.
In terms of equilibrium computation, \cite{grontas2024bighypebestintervention} develops nonsmooth analysis–based algorithms for bilevel games and establishes convergence guarantees.
Alternatively, relaxation-based approaches \cite{doi:10.1137/S1052623499361233} reformulate equilibrium constraints into standard nonlinear programs (NLPs) by progressively driving a relaxation parameter to zero.
While effective for low-dimensional or smooth problem instances, these methods often suffer from scalability limitations and numerical instability in high-dimensional settings involving discontinuous or ambiguous deception mechanisms.
To overcome these challenges, recent advances in hypergradient estimation and implicit differentiation provide a promising direction for scalable equilibrium computation \cite{grontas2024bighypebestintervention}.
However, existing studies are largely restricted to two-player formulations and single-valued BR assumptions.
Extending hypergradient-based methods to three-party games with set-valued BR mappings deserves further investigation.

\section{Three-party Deception Game Model}

In this section, we first present the notations and preliminaries in Section II.A. We then develop a unified deception game model in Section II.B.

\subsection{Notation and Preliminaries}
  \noindent\textbf{Notation}: Let   $\mathbb{N}$ denote the set of non-negative integers, $\mathbb{R}^n$ denote the $n$-dimensional Euclidean space equipped with the standard Euclidean norm $\|\cdot\|$, $\|A\|$ denote the operator norm of the matrix $A$,  $ \operatorname{col}\{x_{1},\dots,x_{n}\} \!  = \!(x^\top_{1}, \dots,x^\top_{n})^\top $, $x_i\in\mathbb{R}^m$, $i=1,\dots,n$.

For a differentiable scalar-valued function $f:\mathbb{R}^m\to \mathbb{R}$, $\nabla f$ denotes its gradient.
For a differentiable vector-valued mapping $F:\mathbb{R}^n \to \mathbb{R}^m$, we denote by {$\mathrm{J}F(x)\in\mathbb{R}^{m\times n}$} the Jacobian of $F$ at $x$.
More generally, for $F:\mathbb{R}^n \times \mathbb{R}^p \to \mathbb{R}^m$, $\mathrm{J}_1F(x,y)\in\mathbb{R}^{m\times n}$ and $\mathrm{J}_2F(x,y)\in\mathbb{R}^{m\times p}$ denote the partial Jacobians of $F$ with respect to its first and second arguments, respectively.
When $m=1$, these reduce to the partial gradients $\nabla_1 F$ and $\nabla_2 F$.

For a point $x\in\mathbb{R}$, let $\delta(x)$ denote a neighborhood of $x$, $\mathring{\delta}(x)$ its punctured neighborhood, and $\delta_{-}(x)$ and $\delta_{+}(x)$ its left and right neighborhoods, respectively.

 \noindent\textbf{Convex Analysis and Operator Theory}:  Let $K \subset \mathbb{R}^n$ be a non-empty closed convex set and $F: K \to \mathbb{R}^n$ be a continuous mapping. The variational inequality problem, denoted as $\text{VI}(K, F)$, is to find $x^* \in K$ such that $\langle F(x^*), x - x^* \rangle \geq 0$ for all $x \in K$. The mapping $F$ is $\mu$-strongly monotone if there exists $\mu > 0$ such that
    $ \langle F(x) - F(y), x - y \rangle \geq \mu\|x - y\|^2, \ \text{for all} \ x, y \in \mathbb{R}^n$. The operator $\mathbb{P}_K(\cdot)$ denotes the orthogonal projection onto the convex and closed set $K$ in Euclidean space, i.e., 
$\mathbb{P}_K(x) = \underset{y\in K}{\mathrm{argmin}} \|x - y\|.$ \\

\noindent\textbf{Nonsmooth Analysis}: The mapping $F$ is $L$-Lipschitz continuous on $K$ if there exists $L > 0$ such that
    $ \|F(x) - F(y)\| \leq L\|x - y\|, \text{for all}\  x, y \in K$. For a locally Lipschitz function $f: \mathbb{R}^n \to \mathbb{R}^m$, the Clarke Jacobian at $x$, denoted by $\partial f(x)$, is the convex hull of the limits of Jacobians at nearby differentiable points: $\partial f(x) = \text{conv} \{ \lim \mathrm{J}f(x_i): x_i \to x, x_i \in \Omega_f \}$, where $\Omega_f$ is the set of points where $f$ is differentiable.  Let $F:\mathbb{R}^n\to\mathbb{R}^m$ be a locally Lipschitz function. A set-valued mapping $\mathcal{J}_F:\mathbb{R}^n \rightrightarrows \mathbb{R}^{m\times n}$ is called a conservative Jacobian of $F$ if it has a closed graph, is locally bounded, and satisfies
$
\frac{d}{dt}F(\rho(t)) \in \{A\dot{\rho}(t):A\in \mathcal{J}_F(\rho(t))\}
\quad \text{for a.e. } t,
$
for every absolutely continuous curve $\rho:\mathbb{R}\to\mathbb{R}^n$.
A locally Lipschitz function is called path differentiable if it admits a conservative Jacobian.
If $F:\mathbb{R}^n\to\mathbb{R}$ is path differentiable and $\mathcal{J}_F$ is a conservative gradient of $F$, then
$
\mathcal{J}_F(x)=\{\nabla F(x)\}
\quad \text{for a.e. } x\in\mathbb{R}^n
$
\cite[Theorem 1]{bolte2021conservative}.

\subsection{Deception Game Model}\label{sec:problem_formulation}

Consider a three-party leader–follower security game, where the defender protects the system against external attacks in the presence of an unidentified insider, who may either cooperate with the defender to support system operation or collude with the attacker for private gain.

The defender makes the first decision and acts as the top-level leader, denoted by $X$. The insider then responds to the defender’s action and subsequently leads the attacker as the middle-level follower, denoted by $Y$. Finally, the attackers make their decisions based on the actions of both the defender and the insider. These attackers constitute the bottom-level followers, denoted by $\boldsymbol{Z}=\{Z_1,Z_2,\ldots,Z_N\}$, where $N$ is the number of bottom-level players and $Z_i$ represents the $i$-th bottom-level follower.

The strategy sets for the players are defined as follows. The top-level player $X$ chooses a strategy $x$ from its set $\Omega_x = [x_{\min}, x_{\max}]$. Similarly, the middle-level player $Y$ chooses $y$ from $\Omega_y = [y_{\min}, y_{\max}]$. Each bottom-level player $Z_i$, for $i=1, \ldots, N$, selects a strategy $z_i$ from the strategy set $\Omega_{z,i} = [z_{i,\min}, z_{i,\max}]$. Define $\boldsymbol{z}=\text{col}\{z_1,z_2,\ldots,z_N\} \in {\Omega}_{\boldsymbol{z}}$ as the collective strategy vector of the bottom-level players, where $z_i \in \Omega_{z,i}, {\Omega}_{\boldsymbol{z}}=\prod_{i=1}^N\Omega_{z,i}\subset \mathbb{R}^N$. Then define  $\boldsymbol{z}_{-i}=\text{col}\{z_1,z_2,\ldots,z_{i-1},z_{i+1},\ldots,z_N\}$ as strategy profile of all bottom-level players except for $Z_i$.

  Define $U_X:\Omega_x\times \Omega_y\times {\Omega}_{\boldsymbol{z}}\to \mathbb{R}$, $U_Y:\Omega_x\times \Omega_y\times {\Omega}_{\boldsymbol{z}}\to \mathbb{R}$, and $U_{z_i}:\Omega_x\times \Omega_y\times {\Omega}_{\boldsymbol{z}}\to \mathbb{R}$ as the utility functions of players $X$, $Y$, and $Z_i$. Let $U_{\boldsymbol{z}}=\{U_{z_1},U_{z_2},\ldots,U_{z_N}\}$. Each player aims to maximize its utility.

A key feature of this game is the introduction of strategic deception, where a leader possesses private knowledge and selectively discloses a manipulated parameter to induce favorable behaviors from followers. {Such scenarios are prevalent in security games involving incomplete information, exemplified by honeypot deception or network topology masking~\cite{pawlick2019game,10.1145/3214305}, where the defender deliberately reveals falsified system states to mislead attackers.}  Here, the top-level player $X$ can manipulate the game environment perceived by the followers. Let $\theta_0\in \mathbb{R}^m$ be the true parameter of the game. Player $X$ can select a deception parameter $\theta$ from a deception set $\Theta \subseteq \mathbb{R}^m$ to alter the followers' perception of the game, with the goal of maximizing its own utility. The followers, $Y$ and $\boldsymbol{Z}$, are unaware of this manipulation.

The leader X strategically selects not only an action $x$ but also a deception parameter $\theta$ from a set $\Theta$ to influence the followers' decisions. The followers, unaware of the deception, perceive $\theta$ as the true parameter of the game. The leader's goal is to choose a pair $(x, \theta)$ that maximizes its own utility $U(x,y,\boldsymbol{z},\theta_0)$,  which gives rise to the game formulation:
\begin{equation}
\noindent{\mathcal{G}(\theta_0,\Theta)=\{\{X,Y,\boldsymbol{Z}\},\Omega\times \Theta, \{U_x,U_y,U_{\boldsymbol{z}}\},\theta_0\}}
\end{equation}
where $\Omega=\Omega_x\times\Omega_y\times\Omega_{\boldsymbol{z}}$.
If the leader cannot choose  the  deception parameter, then  $\Theta=\{\theta\}$ is fixed, that is, 
\begin{equation}
{\mathcal{G}(\theta_0,\theta) = \{\{X, Y, \boldsymbol{Z}\}, \Omega, \{U_x, U_y, U_{\boldsymbol{z}}\}, \{\theta_0,\theta\}\}}
\end{equation}
The classical leader-follower game can be viewed as a special case of our proposed model, which arises when $\Theta=\{\theta_0\}$.

It follows from many security and socio-economic scenarios~\cite{10520318,7357413,9448511} that the game model  can be precisely formulated as follows:
\begin{align}\label{utility}
&\max_{x \in \Omega_x \hphantom{,i}} U_X(x, y, \boldsymbol{z}, \theta_0) = B(x, \theta_0) + f_1(y, \boldsymbol{z}) x\notag\\
&\max_{y \in \Omega_y \hphantom{,i}} U_Y(x, y, \theta) = f_2(x, \theta) y + f_3(x)\\
&\max_{z_i \in \Omega_{z,i}} U_{Z_i}(x, y, z_i, \boldsymbol{z}_{-i}, \theta) = f_{4_i}(x, y, z_i, \boldsymbol{z}_{-i}, \theta)\notag
\end{align}

In this formulation, we explicitly distinguish the information sets, while the leader optimizes based on the true parameter $\theta_0$, the followers $Y$ and $Z_i$ are unaware of the deception and make their decisions based on the manipulated parameter $\theta$.
The specific formulation of the leader's  objective $U_X$ decomposes the leader's total utility into two parts: 
(1)  utility $B(x, \theta_0):\Omega_x\to \mathbb{R}$, representing the  utility from its own 
decision $x$ ; (2) an interaction term $f_1(y, z) x$, 
capturing gains or costs from interactions with followers $Y, \boldsymbol{Z}$, scaled by the leader's 
own effort $x$. 
The utility function \( U_Y \) is linear in \( y \), with its slope and intercept represented solely by functions of \( x \), namely $f_2: \Omega_x \times \Theta \to \mathbb{R}$ and $f_3: \Omega_x \to \mathbb{R}$. The term $f_2(x, \theta) y$ represents variable revenue, depending 
on its own decision $y$ and a price $f_2(x, \theta)$ set by the upper level. 
The term $f_3(x)$ is a fixed cost or benefit independent of $y$. The utility $U_{Z_i}$ is $f_{4_i}(x, y,z_i,\boldsymbol{z}_{-i},\theta)$ with
different mathematical forms in different practical cases. 
Many problems can be formulated by the developed three-party deception game $\mathcal{G}(\theta_0,\Theta)$. For example, in secure wireless communication~\cite{8454822}, the source node, relay, and eavesdroppers act as the leader, middle follower, and bottom followers, respectively; similarly, in defense against IA-FDI attacks~\cite{9863873,liu2021defense}, the defender, insider, and attackers serve as the leader, middle follower, and bottom followers.


Next, we introduce the decision-making scheme of the game model and its equilibrium solutions.

For any given upper-level decisions $(x, y)$ and the leader's manipulated parameter $\theta$, the $N$ bottom-level followers $Z_i$ engage in a simultaneous game. 
Each bottom player $i$ attempts to maximize its utility function $U_{Z_i}(x, y, z_i, z_{-i}, \theta)$. The outcome of this game is an NE \cite{doi:10.1137/1.9781611971132}, denoted as $\boldsymbol{z}^*$, which is  a strategy profile such that no player $i \in \{1, \dots, N\}$ can unilaterally improve their utilities by deviating, formally satisfying
\begin{equation}
  U_{Z_i}(x, y, z_i^*, z_{-i}^*, \theta) \geq U_{Z_i}(x, y, z_i, z_{-i}^*, \theta),\, \forall z_i \in \Omega_{z_i}. 
\end{equation}
These  equilibria constitute the bottom-level BR mapping
\begin{multline}
\mathrm{BR}_{\boldsymbol Z}(x,y,\theta)
=\{\,\boldsymbol z^*| \ U_{Z_i}(x,y,z_i^*,\boldsymbol z_{-i}^*,\theta)\\
\quad \ge U_{Z_i}(x,y,z_i,\boldsymbol z_{-i}^*,\theta),\
\forall i\in\{1,\ldots,N\}, z_i\in\Omega_{z_i}\,\}.
\end{multline}
Given $(x,\theta)$, the middle-level follower $Y$ solves the following optimization problem to maximize $U_Y$,
\[
\max_{y \in \Omega_y} U_Y(x, y, {\theta}).
\]
From the utility function of $Y$, its decision rule is as follows:
\begin{enumerate}[label=(\alph*)]
    \item when $f_2(x, \theta) > 0$, $y = y_{\max}$,
    \item when $f_2(x, \theta) < 0$, $y = y_{\min}$,
    \item when $f_2(x, \theta) = 0$, $y \in [y_{\min}, y_{\max}]$.
\end{enumerate}
We can construct the middle-level BR mapping:
\begin{equation}
  \text{BR}_Y(x, \theta) =
  \begin{cases}
    y_{\max} &f_2(x, \theta) > 0\\
    y_{\min} &f_2(x, \theta) < 0\\
    [y_{\min}, y_{\max}] &f_2(x, \theta) = 0
  \end{cases}
\end{equation}

The leader $X$, positioned at the top of the decision hierarchy, can perfectly anticipate the reactions of $Y$ and $\boldsymbol{Z}$. Consequently, its objective is to select an optimal pair $(x, \theta)$ that maximizes its own utility under the true parameter $\theta_0$,
\begin{equation}
\begin{aligned}
& \max_{x \in \Omega_x,\ \theta \in \Theta} U_X(x, y, \boldsymbol{z}, \theta_0)  \\
& \text{subject to } y \in \text{BR}_Y(x, \theta)  \\
& \phantom{\text{subject to }} \boldsymbol{z} \in \text{BR}_{\boldsymbol{Z}}(x, y, \theta) 
\end{aligned}
\end{equation}
The leader's optimization proceeds in two steps.  First, for a given manipulated parameter $\theta$, it selects an optimal decision $x$. 
Then, among all admissible manipulated parameters, it chooses the one that maximizes its utility.  In such a decision-making framework, the equilibrium solution is referred to as the DSE.

\begin{definition}
    {For a deception parameter set $\Theta$, the tuple $(x^*, y^*, z^*, \theta^*)$ constitutes a {DSE} if
    \begin{equation}
      \begin{aligned}
        &(x^*, \theta^*) \in \underset{x \in \Omega_x, \theta \in \Theta}{\arg\max}   U_X(x, y, z, \theta_0) \\
        & \text{subject to } y \in \text{BR}_Y(x, \theta)  \\
        & \phantom{\text{subject to }} \boldsymbol{z} \in \text{BR}_{\boldsymbol{Z}}(x, y, \theta) 
      \end{aligned}
    \end{equation}
    where $y^* \in \mathrm{BR}_Y(x^*, \theta^*)$ and $z^* \in \mathrm{BR}_{\boldsymbol{Z}}(x^*, y^*, \theta^*)$ are the corresponding equilibrium strategies of the followers.}
\end{definition}

Based on the leader’s assumptions about the follower’s decision-making preferences, we introduce two key subclasses of DSE: SDSE and WDSE.
WDSE represents the leader adopting a pessimistic strategy, assuming the follower will choose the action within their BR set that minimizes the leader's utility. Conversely, SDSE represents the leader adopting an optimistic strategy,
assuming the follower will choose the action within their BR set that maximizes the leader’s utility.
\begin{definition}\label{DSE}
  For a deception parameter set $\Theta$, the tuple $(x^*, y^*, \boldsymbol{z}^*, \theta^*)$ constitutes a {WDSE} if
    \begin{equation} \label{eq:wdse}
      \begin{aligned}
        &(x^*, \theta^*) \in \underset{x \in \Omega_x, \theta \in \Theta}{\arg\max}  \min_y \min_{\boldsymbol{z}}  U_X(x, y, \boldsymbol{z}, \theta_0) \\
        & \text{subject to } y \in \text{BR}_Y(x, \theta)  \\
        & \phantom{\text{subject to }} \boldsymbol{z} \in \text{BR}_{\boldsymbol{Z}}(x, y, \theta) 
      \end{aligned}
    \end{equation}
 and $(y^*, \boldsymbol{z}^*)$ are the specific responses that attain the minimum in \eqref{eq:wdse}.
\end{definition}

Similar to Definition~\ref{DSE},  the tuple $(x^*, y^*, \boldsymbol{z}^*, \theta^*)$ corresponds to an SDSE by replacing the \(\min\) operators in \eqref{eq:wdse} with \(\max\).
In general, the SDSE exists, whereas the WDSE may not. The existence proof for the SDSE and an example demonstrating the nonexistence of the WDSE are provided in Section III. Since WDSE may not exist,  we introduce the concept of $\epsilon$-WDSE \cite{bai2021sampleefficientlearningstackelbergequilibria}.
\begin{definition}\label{epsilon-WDSE}
  {
   A strategy profile $(x^*, y^*, \boldsymbol{z}^*, \theta^*)$  is said to be an $\epsilon$-WDSE  if   for all \(\theta \in \Theta\) and \(x \in \Omega_x\),
  \begin{equation}
    \begin{split}
    &\min_{y\in \text{BR}_Y(x^*,\theta^*)} \min_{\boldsymbol{z}\in \text{BR}_{\boldsymbol{Z}}(x^*,y^*,\theta^*)} U_X(x^*,y,\boldsymbol{z},\theta_0)\geq \\ &\min_{y\in \text{BR}_Y(x,\theta)} \min_{\boldsymbol{z}\in \text{BR}_{\boldsymbol{Z}}(x,y,\theta)} U_X(x,y,\boldsymbol{z},\theta_0)-\epsilon
    \end{split}
  \end{equation}
  where the constant $\epsilon>0$.}
\end{definition}

When \(\Theta = \{\theta\}\), i.e., the leader cannot alter the deception parameter \(\theta\), the equilibrium to this game is a misperception Stackelberg Equilibrium \cite{Cheng_2022}. When there is no deception parameter in the entire game, i.e., \(\Theta = \emptyset\), the equilibrium of the game reduces to a standard Stackelberg equilibrium \cite{10520318}. 

On the other hand, due to deception, the problem can also be formulated as a hypergame  with HNE \cite{Sasaki_2008},
where each player selects the optimal strategy according to their subjective perception of the game structure and the opponents' actions. The information asymmetry here is primarily characterized by the followers optimizing their strategies with respect to the manipulated parameter $\theta$, while the followers assume that all other participants, including the leader, are also optimizing their strategies with respect to $\theta$. The leader, however, possesses knowledge of the true parameter $\theta_0$ and is fully cognizant of the followers' optimization conducted under the manipulated parameter $ \theta$.

\begin{definition}
For the three-level leader-follower game $\mathcal{G}$ with deception parameter $\theta$, 
a strategy profile $({x}^{\diamond}, y^{\diamond}, {\boldsymbol{z}}^{\diamond})$ is said to be an HNE if
\begin{align*}
x^{\diamond} &\in \arg\max_{x \in \Omega_x} U_X(x, y^{\diamond}, {\boldsymbol{z}}^{\diamond},\theta_0),\\
y^{\diamond} &\in \arg\max_{y \in \Omega_y} U_Y(x^{\diamond}, y, {\boldsymbol{z}}^{\diamond},\theta),\\
{z}_{i}^{\diamond} &\in \arg\max_{z_i \in \Omega_{z,i}} U_{Z_i}(x^{\diamond}, y^{\diamond}, z_i, {\boldsymbol{z}}_{-i}^{\diamond},\theta), \quad \forall i \in \{1, \dots, N\}.
\end{align*}
\end{definition}

 We make the following basic assumptions~\cite{10520318,7357413}.
\begin{assumption}
    \label{a 2.1}
 {The utility functions $U_X(x,y,z,\theta_0)$, $U_Y(x,y,z,\theta)$, and $U_{Z_i}(x,y,z_i,z_{-i},\theta)$ are concave in their respective decision variables $x$, $y$, and $z_i$, and are continuously differentiable. The set of deception parameters $\Theta$ is finite and $|\Theta|=q$. }
\end{assumption}

Under Assumption~\ref{a 2.1}, finding an NE strategy profile $\boldsymbol{z}^* \in \Omega_Z$ is equivalent to solving the parametrized variational inequality $\mathrm{VI}(F(x,y,\cdot,\theta), \Omega_{Z})$ \cite{facchinei2003finite}, where the pseudogradient (PG) mapping $F$ is the stacked gradient vector of the followers' utility functions
\begin{equation}
  {
  F(x,y,\boldsymbol{z},\theta) \triangleq \text{col}\{-\nabla_{z_i} f_{4_i}(x,y,\boldsymbol{z},\theta) \}_{i \in N}.}
\end{equation}

The following assumption guarantees the uniqueness of the lower-level equilibrium, i.e., $\mathrm{BR}_{\boldsymbol{Z}}(x,y,\theta)$ is single-valued.
\begin{assumption}
\label{a 2.2}
  {For fixed $x,y$ and $\theta$, PG mapping $F(x,y,\cdot,\theta)$ is $\mu$-strongly monotone and $\kappa$-Lipschitz continuous. The mapping $F(x,y,\boldsymbol{z},\theta)$ is continuously differentiable for any fixed $\theta\in \Theta$.}
\end{assumption}


To guarantee the existence and computability of DSE, the following assumptions are introduced, which were widely used in \cite{grontas2024bighypebestintervention,10520318}.
\begin{assumption}
  \label{a 2.3}
  \begin{enumerate}
    \item $\nabla_x B(x,\theta_0), \nabla_z f_1(y,z)$ is Lipschitz continuous.
  \item $f_2(x,\theta)$ has finite zero points.
  \item 
The PG mapping $F$ is definable\footnote{Definable functions form a broad class that includes most functions used in optimization and machine learning, such as semialgebraic functions, as well as functions involving exponentials and logarithms. Definable functions are closed under standard operations (e.g., addition, multiplication, composition) and possess desirable properties such as path differentiability \cite{bolte2021nonsmooth,van1996geometric}.}  and there exists a constant $L_{F}$ that satisfies
  \begin{align*}
    \|\mathrm{J}_1 F(x,y,\boldsymbol{z})-\mathrm{J}_1 F(x,y,\hat{\boldsymbol{z}})\|\leq L_{F}\|\boldsymbol{z}-\hat{\boldsymbol{z}}\|\\  
    \|\mathrm{J}_3 F(x,y,\boldsymbol{z})-\mathrm{J}_3 F(x,y,\hat{\boldsymbol{z}})\|\leq L_{F}\|\boldsymbol{z}-\hat{\boldsymbol{z}}\|
  \end{align*}
  for any $\boldsymbol{z},\hat{\boldsymbol{z}} \in \Omega_{\boldsymbol{z}}$.

  \end{enumerate}
\end{assumption}

For a given $\theta$, $f_2(\cdot,\theta)$ has a finite number of zeros on $\Omega_x = [x_{\min}, x_{\max}]$, denoted by $\mathcal{Z}_{f_2}(\theta) = \{x_1, x_2, \ldots, x_{q_\theta}\}$ with $x_1 < x_2 < \dots < x_{q_\theta}$. With $x_0 \triangleq x_{\min}$ and $x_{q_\theta+1} \triangleq x_{\max}$, we can partition $\Omega_x$ into $q_\theta + 1$ closed sub-intervals $\Omega_{x,i} \triangleq [x_{i-1}, x_i]$, such that $\Omega_x = \bigcup\limits_{i=1}^{q_\theta+1} \Omega_{x,i}$. By construction, $f_2(x,\theta)$ does not change sign on any sub-interval $\Omega_{x,i}$. It is either less than or equal to 0 or greater than or equal to 0. Thus, $\text{BR}_Y(x,\theta)\equiv y_{\max}$ or $\text{BR}_Y(x,\theta)\equiv y_{\min}$.
On each interval $\Omega_{x,i}$, our problem can be reformulated as
\begin{equation}
  \max_{x \in \Omega_{x,i}}\hat{U}_X(x)=\max_{x \in \Omega_{x,i}}\ U_X(x, BR_Y(x,\theta) BR_{\boldsymbol{Z}}(x,y,\theta), \theta_0)
\end{equation}
Noting that the function $\hat{U}_X$ is piecewise continuous on $\Omega_x$, we will, therefore,  work over the interval \(\Omega_{x,i}\) when discussing the DSE.

In our model, the top-level leader can influence the decision environment of the lower-level followers by announcing a strategically chosen deception parameter $\theta$.
The DSE corresponds to the leader adopting an optimal deception strategy under a hierarchical leader–follower structure.
However, a standard DSE relies on the followers’ strict adherence to this hierarchy, an assumption that may be fragile in practice.
A more robust and desirable equilibrium arises when the leader’s optimal strategy under the hierarchical model also coincides with its optimal choice in a simultaneous-move game.

Therefore, this paper addresses the following two problems:
\begin{enumerate}
    \item Providing the conditions under which a WDSE or an SDSE is consistent with an HNE.
    \item Developing efficient algorithms to compute a WDSE (including $\epsilon$-WDSE) and an SDSE.
\end{enumerate}

\section{Existence and Consistency of Equilibria}

In this section, we first establish the existence of SDSE, WDSE, and HNE, and then explore the consistency between these equilibria.
\subsection{Equilibrium existence}
\begin{lemma}
  \label{lemma exsit}
  {
   Under Assumptions~\ref{a 2.1}-\ref{a 2.3},
   the SDSE set is nonempty. Thus, the DSE set is nonempty.}
\end{lemma}
\begin{proof}
  For a fixed parameter $\theta$, $\text{BR}_Y(x,\theta)$ is upper semicontinuous. 
  We now proceed to prove that 
  \begin{equation}
      f_{\theta}(x)=\max_{y\in \text{BR}_Y(x,\theta)}U_X(x,\text{BR}_Y(x),\text{BR}_{\boldsymbol{Z}}(x,\text{BR}_Y(x),\theta))
  \end{equation}
 is upper semicontinuous. 
 
 For any sequence $\{x_n\}$ and $x_n\to x_0$, let $f_{\theta}(x_0)=U_X(x_0,y_0,\text{BR}_{\boldsymbol{Z}}(x_0,y_0,\theta))$.
  Let  $y_n \in \text{BR}_Y(x_n,\theta)$ and $f_{\theta}(x_n)=U_X(x_n,y_n,\text{BR}_Z(x_n,y_n,\theta))$. Consider a convergent subsequence \( \{y_{n_k}\} \) of \( \{y_n\} \), whose limit is \( y^* \). Since $\text{BR}_Y(x,\theta)$ is upper semicontinuous, $y^*\in \text{BR}_Y(x_0,\theta)$. From the definition of $y_n$, we have $f_{\theta}(x_{n_k})= U_x(x_{n_k},y_{n_k},\text{BR}_{\boldsymbol{Z}}(x_{n_k},y_{n_k},\theta))$. Then
 \begin{equation}
  \limsup f_{\theta}(x_n)= \limsup U_X(x_n,y_n,\text{BR}_{\boldsymbol{Z}}(x_n,y_n,\theta))
 \end{equation}
 Because $U_X$ and $\text{BR}_{\boldsymbol{Z}}$ are continuous, 
 \begin{equation}
  \begin{aligned}
  \limsup f_{\theta}(x_n)&= U_X(x_0,y^*,\text{BR}_{\boldsymbol{Z}}(x_0,y^*,\theta))\\
  f_{\theta}(x_0)&\geq U_X(x_0,y^*,\text{BR}_{\boldsymbol{Z}}(x_0,y^*,\theta))
  \end{aligned}
 \end{equation}

 Therefore, $f_{\theta}(x)$ is upper semicontinuous. Moreover, since $\Omega_x$  is a compact set, there exists a point $x_{\theta}$ such that $f_{\theta}(x)$ attains its maximum at $x_{\theta}$. 
 Among all pairs $(x_{\theta}, \theta)$, let $(x^*,\theta^*)$ be one that maximizes the leader's utility. Then
 \begin{equation}
 y^* \in \underset{y \in \text{BR}_Y(x^*, \theta^*)}{\arg\max} U_X(x^*, y, \text{BR}_{\boldsymbol{Z}}(x^*, y, \theta^*)).
 \end{equation} 
 \begin{equation}
  \boldsymbol{z}^*=\text{BR}_{\boldsymbol{Z}}(x^*, y^*, \theta^*)
 \end{equation}
 This pair $(x^*,y^*, \boldsymbol{z}^*,\theta^*)$ constitutes the SDSE.
\end{proof}

The nonexistence of the WDSE from the discontinuity of the set-valued mapping \( \text{BR}_Y(x) \). We provide an example as follows.
Let
\begin{equation*}
  \bar{U}_X(x,\theta)=\min_{y\in \text{BR}_Y(x,\theta)} \min_{\boldsymbol{z}\in \text{BR}_{\boldsymbol{Z}}(x,y,\theta)} U_X(x,y,\boldsymbol{z})
\end{equation*}
Then the tuple $(x^*,y^*,\boldsymbol{z}^*,\theta^*)$ is a WDSE such that $x^*\in\arg\max \bar{U}_X(x,\theta)$. It should be noted that, in many cases, $\bar{U}_X(x,\theta)$ does not attain a maximum value. Let
\begin{align*}
  &U_X(x,y,\boldsymbol{z})=-2x^2+2x-yx\\
  &U_Y(x,y,\theta)=(x-1)(x-2)y\\
  &x\in [0,3], y\in[-1,1],\Theta=\emptyset
\end{align*} 
Then \begin{equation}
  \bar{U}_X(x,\theta)=\begin{cases}
  -2x^2+x, \quad &x\leq 1\\
  -2x^2+3x, \quad &x\in(1,2)\\
  -2x^2+x, \quad &x\geq 2
  \end{cases}
\end{equation}
It is clear that $\bar{U}_X(x,\theta)$ has no maximum value,
 as illustrated in Fig.~\ref{figexample 1}. Fortunately, for any $\epsilon > 0$, the $\epsilon$-WDSE is guaranteed to exist. 
 {We establish the existence of the $\epsilon$-WDSE in the following and discuss the conditions under which the WDSE exists.}
\begin{figure}[t]

\centering
\includegraphics[width=0.9\columnwidth]{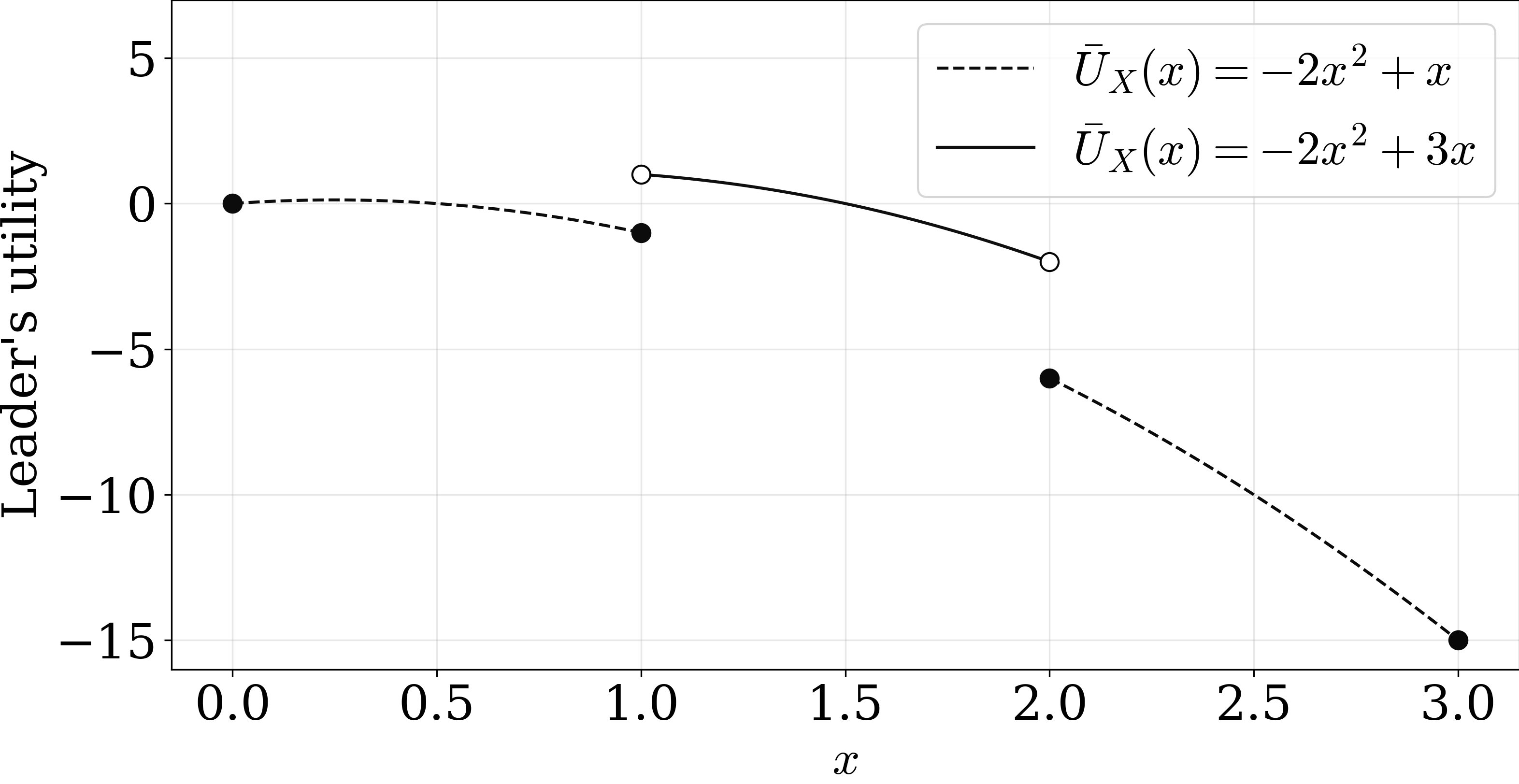}
\captionsetup{justification=centering}
\caption{The leader's utility}
 \label{figexample 1}
\end{figure}
\begin{lemma} 
    \label{lemmaa 2}
   Under Assumptions~\ref{a 2.1}-\ref{a 2.3}, for any $\epsilon>0$ there exists an $\epsilon$-WDSE in game $\mathcal{G}(\theta_0,\Theta)$.
   Let $\{\epsilon_n\}$ be a sequence of positive scalars converging to $0$, and $\{(x_n, \theta_n)\}$ be a corresponding sequence of strategies such that each $(x_n, \theta_n)$ is an $\epsilon_n$-WDSE. Then every limit point $(x^*, \theta^*)$ of $\{(x_n, \theta_n)\}$ is an exact WDSE if and only if satisfying the condition $f_2(x^*, \theta^*) \neq 0$ or $\bar{U}_X(x,\theta^*)$ is upper semicontinuous at $x^*$.
\end{lemma}
\begin{proof}
  From the Definition~\ref{epsilon-WDSE}, the existence of an $\epsilon$-WDSE follows directly.
  We now proceed to prove the necessary and sufficient condition under which a limit point of a sequence of \(\varepsilon\)-WDSEs, as \(\varepsilon \to 0\), is itself a WDSE. 

  When the tuple $(x^*,y^*,\boldsymbol{z}^*,\theta^*)$ is a WDSE strategy, if $f_2(x^*,\theta^*)\neq 0$, then the proof is complete. If $f_2(x^*,\theta^*)=0$, then $\text{BR}_Y(x,\theta)=\Omega_y$, thus, 
  \begin{equation*}
      \inf_{y\in \Omega_y} U_X(x^*,y,\text{BR}_{\boldsymbol{Z}}(x^*,y,\theta^*))\geq \underset{x\to x^*} {\lim\sup}\  \bar{U}_X(x,\theta^*).
  \end{equation*}
Then $\bar{U}_X(x,\theta^*)$ is upper semicontinuous at $x^*$.

  Consider a tuple $(x_0,y_0,\boldsymbol{z}_0,\theta_0)$ and $(x_0,\theta_0)$ be a limit point of a $\epsilon$-WDSE sequence as $\epsilon\to 0$. When $f_2(x_0,\theta_0)\neq 0$, we obtain $\bar{U}_X(x,\theta_0)$ is continuous at $x_0$ Then \begin{equation*}
    x_0=\arg \max \bar{U}_X(x,\theta_0).
  \end{equation*}
  When $f_2(x_0,\theta_0)=0$, $\bar{U}_X(x,\theta_0)$ is upper semicontinuous at $x_0$. Note that
  \begin{equation*}
      \lim_{n\to \infty} \bar{U}_X(x_n,\theta_n)=\sup_{x\in \Omega_x, \theta\in \Theta} \bar{U}_X(x,\theta)
  \end{equation*}
  Then
  \begin{equation*}
      \bar{U}_X(x_0,\theta_0)\geq \lim_{n\to \infty} \bar{U}_X(x_n,\theta_n)=\sup_{x\in \Omega_x, \theta\in \Theta} \bar{U}_X(x,\theta)
  \end{equation*}
  Thus, $(x_0,y_0,\boldsymbol{z}_0,\theta_0)$ is a WDSE.
\end{proof}

Regarding the existence of the HNE, we rely on results from \cite{facchinei2003finite}.
\begin{lemma}
  \label{lemmaa HNE}
   Under Assumptions~\ref{a 2.1}-\ref{a 2.3}, there exists an HNE in game $\mathcal{G}(\theta_0,\theta)$ for any $\theta \in \Theta$.
\end{lemma}

\subsection{Consistency between DSE and HNE}

The motivation for aligning a DSE with an HNE stems from the fact that the leader’s action \( x \) may be difficult to observe accurately. This can lead followers to act without observing the leader’s move, leaving the leader uncertain whether the followers are playing according to a DSE or an HNE strategy, resulting in a reduction of the leader’s utility.

    

\begin{remark}
    Although the deception parameter \( \theta \) is also a decision variable under the leader’s control, its observability differs from that of the leader’s action \( x \). In many practical applications, $\theta$ functions as a public signal designed for dissemination, whereas $x$ represents an internal operational variable that is often opaque or costly to monitor.
\end{remark}

For a WDSE or an SDSE strategy $(x^*,y^*,\boldsymbol{z}^*,\theta^*)$, define
\begin{equation}\label{zero_pay}
\begin{aligned}
T_1(x)&=\frac{\partial U_X}{\partial x}(x,y^*,\boldsymbol{z}^*),\\
g(x,\theta)&=
\begin{cases}
\inf\limits_{y\in \mathrm{BR}_Y(x,\theta)} U_X(x,y,\mathrm{BR}_{\boldsymbol{Z}}(x,y,\theta)), & \text{for WDSE},\\[1mm]
\sup\limits_{y\in \mathrm{BR}_Y(x,\theta)} U_X(x,y,\mathrm{BR}_{\boldsymbol{Z}}(x,y,\theta)), & \text{for SDSE}.
\end{cases}
\end{aligned}
\end{equation}

Define the utility function of the leader under the leader–follower  scheme as
\begin{equation}
  \tilde{U}_X(x,\theta)=
  \begin{cases}
    \hat{U}_X(x,\theta) \quad &x\in(x_{i-1},x_i)\\
    g(x_i,\theta) \quad & x_i\in \{x_1,x_2,\ldots,x_{q_{\theta}}\}
  \end{cases}
\end{equation}
{where} $\hat{U}(x,\theta)=U_X(x,\text{BR}_Y(x,\theta),\text{BR}_{\boldsymbol{Z}}(x,\text{BR}_Y(x,\theta),\theta))$. 
Since \( \hat{U}_X(x,\theta) \) is piecewise smooth \cite{grontas2024bighypebestintervention}, \( \tilde{U}_X(x,\theta) \) is also piecewise smooth. Thus, when $\tilde{U}_X$ is differentiable at $x$, we define
\begin{equation}
  T_2(x,\theta)= \frac{\partial \tilde{U}_X}{\partial x}(x,\theta)
\end{equation}


We now present a necessary and sufficient condition under which  WDSE or SDSE coincides with an HNE. The proof is given in Appendix~\ref{consistencesection}.


\begin{theorem}\label{consistence}
Under Assumptions~\ref{a 2.1}--\ref{a 2.3}, when for any $x\in\Omega_x$, there exist $\delta_-(x)$ and $\delta_+(x)$ such that $\tilde{U}_X(x,\theta)$ is monotone on $\delta_-(x)$ and $\delta_+(x)$, any WDSE or SDSE $(x^*,y^*,\boldsymbol{z}^*,\theta^*)$ is an HNE if and only if at least one of the following conditions holds:
\begin{enumerate}
    \item $T_1(x^*)=0$;
    \item $\tilde{U}_X(x^*,\theta^*)$ is not differentiable and there exists a $\mathring{\delta}(x^*)$ such that
    $T_1(x)\cdot T_2(x,\theta^*)>0$ for all $x\in\mathring{\delta}(x^*)$;
    \item $\tilde{U}_X(x^*,\theta^*)$ is differentiable and there exists a $\delta(x^*)$ such that
    $T_1(x)\cdot T_2(x,\theta^*)>0$ for all $x\in\delta(x^*)$.
\end{enumerate}
\end{theorem}
{Theorem~\ref{consistence} not only provides a method for verifying whether a DSE (i.e., WDSE, SDSE) is an HNE, but also offers a way to find a DSE consistent with a given HNE. In terms of computational complexity, the theorem involves only local computations related to the DSE, making it computationally efficient and straightforward to implement.}

When multiple DSE exist, the leader can adjust $\theta$ to ensure the robustness of their utility. For example,
let $U_X(x,y,z)=-(x-1)^2-(z-2)^2+25$, $U_Y(x,y)=(|x|+1)y$ and $U_Z(x,y,z,\theta)=-(z-\theta x)^2$. Set $x\in[-2,2]$, $y\in[-2,2]$, $z\in[-2,2]$ and $\Theta=\{0,-\frac{4}{3}\}$. Then for $\theta_1=0$ and $\theta_2=-4/3$, the profiles $(x, y, z,\theta) = (1, 2,0, 0)$ and $(-0.6,2, 0.8, -4/3)$, both constitute DSE. However, the DSE at $\theta=0$ exhibits robustness as the DSE coincides with the HNE. Fig.~\ref{fig:robustness}  illustrates our conclusion.
\begin{figure}[htbp]
    \centering
    \includegraphics[width=0.8\linewidth]{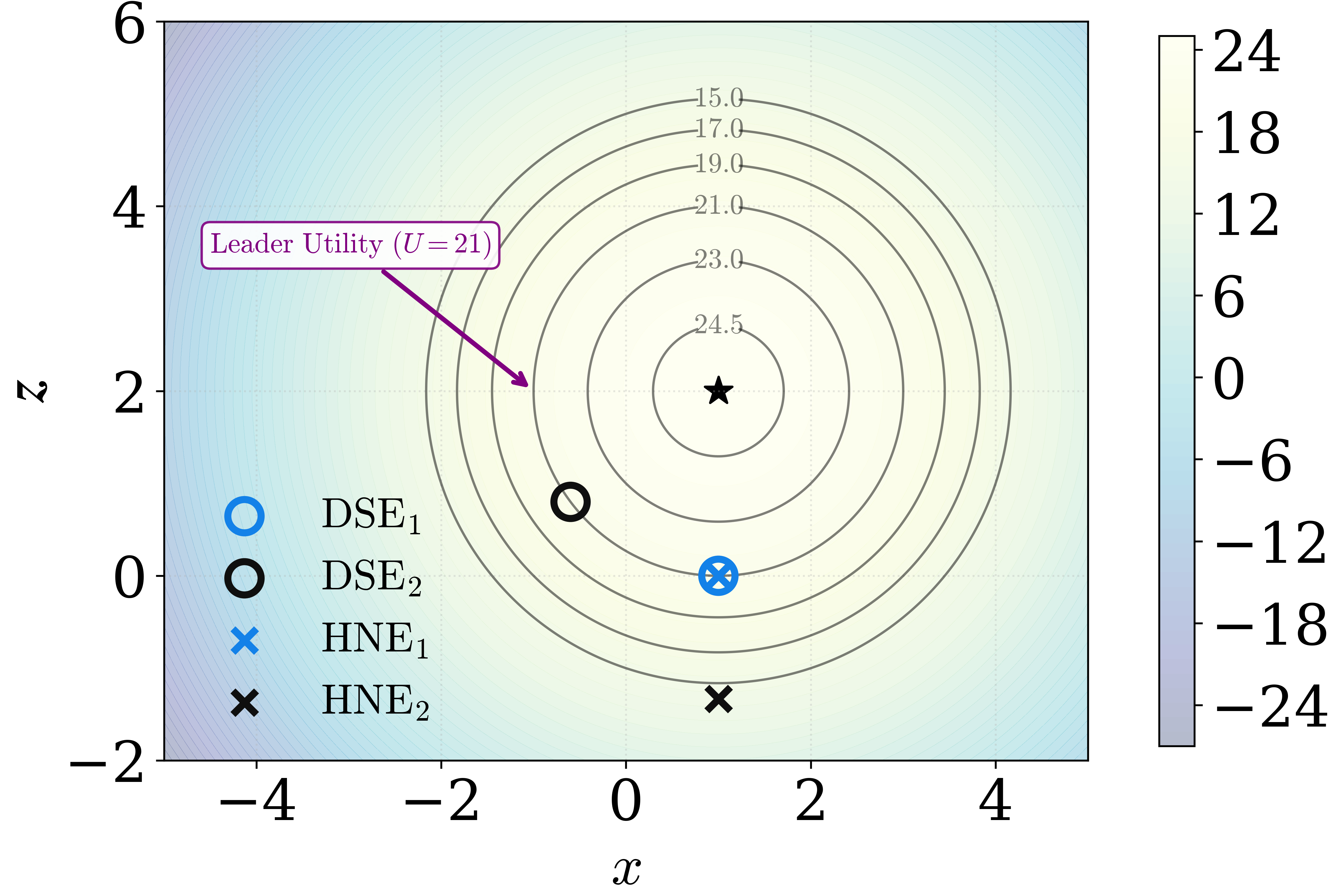}
    \caption{Robustness Assurance}
    \label{fig:robustness}
\end{figure}


\section{Algorithm design for DSE seeking}
Following the previous section, we present the method for computing  WDSE (SDSE). The complete algorithm is summarized in Algorithm 1. For a fixed deception parameter $\theta$, we traverse all intervals $\Omega_{x,i}$ and perform projected gradient ascent using the hypergradient:
\begin{equation}\label{hypergradient}
    \nabla \hat{U}_X^k \!=\! \nabla_1 \hat{U}_X(x^k, {y^{k+1}}\!, {z^{k+1}}\!) \! +\!  ({s^{k+1}}\! )^\top\! \nabla_3 \hat{U}_X(x^k, {y^{k+1}}\!, {z^{k+1}}\!).
\end{equation}
{where the followers' strategies $y^{k+1}$ and $z^{k+1}$ are obtained by Algorithm 2 given fixed $x^k$.}
 To determine the optimal strategy $x^*(\theta)$ for the fixed $\theta$, we compare the utilities at the limiting solutions of these intervals with the leader's utility at the zero points of $f_2$. Specifically, at any zero point $x_{i}$ for $i\in\{1,\dots,x_{q_{\theta}}\}$, the leader's utility is evaluated based on the specific equilibrium concept. Define an infimum under a pessimistic attitude (corresponding to WDSE) or a supremum under an optimistic attitude (corresponding to SDSE)
\begin{equation}
\label{eq:zero_utility}
U_{zero}(x_z) = 
\begin{cases} 
\inf\limits_{y\in \Omega_y} U_X(x_{z},y,BR_{\boldsymbol{z}}(x_{z},y,\theta)), & \text{for WDSE } \\
\sup\limits_{y\in \Omega_y} U_X(x_{z},y,BR_{\boldsymbol{z}}(x_{z},y,\theta)), & \text{for SDSE }
\end{cases}
\end{equation}
The candidate strategy yielding the global maximum utility among all intervals and zero points is then selected as the optimal response for the current $\theta$.

 Finally, by iterating this process over the parameter space $\Theta$, we update the global maximum utility $U^*$ and record the corresponding optimal pair $(x^*,\theta^*)$. 
 
When $\theta$ is fixed, we omit the explicit dependence on $\theta$ and $\theta_0$ for notational brevity, denoting $\text{BR}_y(x,\theta)$ as $\phi_1(x)$ and $\text{BR}_{\boldsymbol{Z}}(x,\phi_1(x),\theta)$ as $\phi_2(x)$.
\begin{algorithm}[t]
\caption{ Hyper-gradient-based algorithm for DSE seeking}
\label{alg:global_wdse}
\begin{algorithmic}[1]
\Require $\Theta$, step sizes $\{\alpha^k\}$, tolerance $\sigma^k$, $U^*\leftarrow \infty$, $x^*\gets 0$

 \;\,\,$U\leftarrow \infty$
\For{$\theta \in \Theta$}
    \State Partition $\Omega_x$ into intervals $\Omega_{x,i}$ and  $\mathcal{Z}_{f_2}$ .
    \For{each $\Omega_{x,i}$ 
    }
        \State Initialize $x^0 \in \Omega_{x,i}, y^0, z^0, s^0, k \leftarrow 0$.
        {  
        \Repeat
            \State $(y^{k+1}, \boldsymbol{z}^{k+1}, s^{k+1}) \leftarrow $
           
           \quad\quad\quad \quad\quad\quad
           \textbf{Inner Loop} $(x^{k},{y^k},\boldsymbol{z}^k,s^k,\sigma^k)$
        \State $x^{k+1} \leftarrow \mathbb{P}_{{\Omega_{x,i}}} [ x^k + \alpha^k {\nabla}\hat{U}_X^k ]$.
       
           
            \State $k \leftarrow k + 1$.
        \Until{convergence} }
         \State $U
\leftarrow
\max\Big\{
\hat{U}_X(x^k,\theta),\;
\max_{i=1,\dots,q_{\theta}} g(x_i,\theta)
\Big\}$ 
         \State $x
\gets
\arg\max\Big(
\{\hat{U}_X(x^k,\theta)\}
\cup
\{g(x_i,\theta)\}_{i=1}^{q_\theta}
\Big)$
    \EndFor
   \If{$U > U^*$}
            \State $U^*\gets U$,  $\theta^*\gets \theta$, $x^*\gets x$
        \EndIf

\EndFor

\State \textbf{Return}
 $U^*,x^*, \theta^*$
\end{algorithmic}
\end{algorithm}

\begin{assumption}
  \label{concave}
  {
For any fixed $y$ and $\theta$, $\hat{U}_X(x,\theta)=U_X(x,y,\text{BR}_{\boldsymbol{Z}}(x,y,\theta),\theta_0)$ is  concave in any interval $\Omega_{x,i}$.}
\end{assumption}

  Assumption~\ref{concave} is intended to prevent the leader's utility function from being non-concave in the leader-follower setting  \cite{doi:10.1137/1.9781611971132}. An intuitive example arises when \( U_{{X}} \) is jointly concave and the utility functions of the bottom-level followers are of linear-quadratic form, in which case the induced function $\hat{U}_X$ admits Assumption~\ref{concave}. Even without this condition, our algorithm can still converge to a composite critical point of the leader's utility function in the leader-follower scheme.

  Under the box constraint $\Omega_{\boldsymbol z}=[\boldsymbol z_{\min},\boldsymbol z_{\max}]$,  the projection admits a piecewise-linear expression 
\begin{equation}
  (\mathbb{P}_{\Omega_{\boldsymbol{z}}}(z))_i=\max(z_{i,\min},\min(z_{i,\max},z_i))
\end{equation}
This solution is piecewise affine. In other words, the space $\mathbb{R}^N$ is partitioned into finitely many hyper-rectangular regions $A_i,\ i\in N_p=\{1,2,\ldots,P\}$, aligned with coordinate axes. Within each region, the projector is a simple affine function and hence continuously differentiable; its differentiability fails only on the boundaries of these regions. More formally, let $\omega(x,y,\boldsymbol{z})=\boldsymbol{z}-\gamma F(x,y,\boldsymbol{z})$ denote the PG step mapping.  Define $\mathcal{P}(x)=\{i\in N_p|\omega(x,\phi_1(x),\phi_2(x))\in A_i\}$ as the set of active region indices at $\omega(x,\phi_1(x),\phi_2(x))$. Let $\delta(x)$ be the radius of the largest ball whose elements are all included in one of the active partitions,
i.e., $\delta(x) := \max\{r \in \mathbb{R}_{\ge 0} \,|\,
{B}(\omega(x, \phi_1(x),\phi_2(x)), r) \subseteq \bigcup_{i \in \mathcal{P}(x)}{A_i}\}$.

\begin{assumption}
\label{delta}
There exists $k' \in \mathbb{N}$ such that $\sigma^k < \delta(x^k)$,
for all $k > k'$.
\end{assumption}

Assumption~\ref{delta} ensures that the estimate \( s \) obtained from Algorithm 2 and the true value \( \mathcal{J}\phi_2(x) \) lie within the same differentiability slice. Moreover, it is readily satisfied because, under box constraints, each slice admits an explicit analytical characterization, allowing us to compute the radius \( r \) of the slice directly.

\subsection{Inner Loop}
The specific steps of the inner loop are outlined in Algorithm 2. Its primary objective is to solve the follower's optimization problem given a fixed leader's variable $x^k$. This process involves three main tasks: (1) determine the sign of $f_2(x,\theta)$ at $\Omega_{x,i}$. (2) approximating the bottom-level follower's BR $\phi_2(x)$, and (3) learning the sensitivity of this response with respect to the leader's variables $\mathrm{J}\phi_2(x)$. The output of this loop provides the necessary information for the leader to perform the informed "projected hypergradient" update \eqref{hypergradient}.

For given $x$ and $y$, the inner loop iteratively computes the follower's optimal strategy $z$. This is achieved through a fixed-point iteration defined by the function $h(x, y, z)$. The update rule is given by
\begin{equation}\label{phi_2}
    \boldsymbol{z}^{\ell+1} = h(x, y, \boldsymbol{z}^\ell),\quad \forall \ell\in \mathbb{N}.
\end{equation}
Define
   $ h(x, y, \boldsymbol{z}) = \mathbb{P}_{\Omega_{\boldsymbol{z}}} \left[\boldsymbol{z} - \gamma F(x, y, \boldsymbol{z})\right]$,
where $\gamma$ is the step size,  $\mathbb{P}_{\Omega_{\boldsymbol{z}}}$ is the projection operator that ensures the updated strategy $\boldsymbol{z}^{l+1}$ remains within the feasible set $\Omega_{\boldsymbol{z}}$. Therefore, $$\boldsymbol{z}=\phi_2(x) \Leftrightarrow \boldsymbol{z}=h(x,\phi_1(x),\boldsymbol{z}).$$
\begin{lemma}
  \label{contraction}
  Under Assumptions~\ref{a 2.1}-\ref{a 2.3}, 
  if $\gamma < \frac{2 \mu}{\kappa^2}$, then $h(x,y,\boldsymbol{z})$ is a contraction mapping with respect to $\boldsymbol{z}$, with the contraction constant $\eta=\sqrt{1-\gamma(2\mu-\gamma \kappa^2)}$.
\end{lemma}

The proof of Lemma \ref{contraction} is provided in Appendix \ref{algorithmsection}.
Therefore, the sequence generated by \eqref{phi_2} converges linearly to $\phi_2(x)$ with rate $\eta$. Since $\phi_2(x)=h(x,\phi_1(x),\phi_2(x))$,we differentiate both sides with respect to $x$ and obtain
\begin{equation}\label{Jac}
  \mathrm{J}\phi_2(x)=\mathrm{J}_1 h + \mathrm{J}_3 h  \mathrm{J} \phi_2(x)
\end{equation}
The absence of $\mathrm{J}_2$ in the above expression is because $\phi_1(x)$ is constant over the interval $\Omega_{x, i}$. The solution of \eqref{Jac} can be obtained via a fixed-point iteration:
\begin{equation}\label{fixed_point}
  \hat{s}^{\ell+1}=\mathrm{J}_1h(x,\phi_1(x),\phi_2(x))+\mathrm{J}_3h(x,\phi_1(x),\phi_2(x)){\hat{s}^\ell}
\end{equation}
{  A direct implementation of \eqref{fixed_point} requires the exact solution $\phi_2(x)$, which in turn demands infinitely many PPG iterations in \eqref{phi_2}, which is an infeasible requirement in practice. To address this, we propose an online approximation scheme that uses the most recent iterate from \eqref{phi_2} as a surrogate for $\phi_2(x)$. }Specifically,
\begin{equation}\label{fixed_point_online}
    \tilde{s}^{\ell+1} = \mathrm{J}_1 h(x, \phi_1(x), \tilde{\boldsymbol{z}}^{\ell+1}) +  \mathrm{J}_3 h(x, \phi_1(x), \tilde{\boldsymbol{z}}^{\ell+1})\tilde{s}^\ell
\end{equation}
This iterative scheme does not require $\phi_2(x)$. It only uses the most recently updated $\tilde{\boldsymbol{z}}^{\ell}$. However, this is not a fixed-point iteration, as the value of $\tilde{\boldsymbol{z}}^\ell$ changes at every iteration step.

For the two iterative schemes in \eqref{fixed_point} and \eqref{fixed_point_online}, we show that the proposed algorithm converges to the true Jacobian. The proof is given in Appendix~\ref{algorithmsection}.

\begin{lemma}
    \label{offline_lemma}
    Under Assumptions~\ref{a 2.1}-\ref{a 2.3}, 
  for a fixed \( x \), if \( \mathcal{P}(x) \) is single-valued, then \( \phi_2(x) \) is differentiable at \( x \), and the sequence generated \eqref{fixed_point} converges to \( \mathrm{J} \phi_2(x) \).
\end{lemma}

Next, we extend the convergence results to the online estimation-based iterative scheme \eqref{fixed_point_online}. We first bound the error between the online-estimated sequence \( \tilde{s}^\ell \) and the sequence \( \hat{s}^\ell \) generated by the fixed-point iteration. Then by further bounding the error between \( \hat{s}^\ell \) and \( \mathrm{J}\phi_2(x) \), we obtain the overall error between the online-estimated sequence \( \tilde{s}^\ell \) and \( \mathrm{J}\phi_2(x) \). This error bound converges to zero as $\ell\to\infty$, therefore, $\tilde{s}^\ell$ converges to $\mathrm{J}\phi_2(x)$ as $\ell\to\infty$.

\begin{lemma}
\label{online_lemma}
  Under Assumptions~\ref{a 2.1}-\ref{a 2.3}, fix $x \in \Omega_{x,i}$ and assume that $\mathcal{P}(x)$ is a singleton.
Then the sequence $\{\tilde{s}^{\ell}\}_{\ell \in \mathbb{N}}$
generated by \eqref{fixed_point_online} converges to $\mathrm{J}\phi_2(x)$.
\end{lemma}
\begin{algorithm}[t]
\caption{\textbf{Inner Loop}}
\label{alg:inner}
\begin{algorithmic}[1]
\Require step sizes $\gamma$, contraction constant $\eta$.
\State \textbf{Input:} $x, {y}, z, s, \sigma$
\State \textbf{Phase 1: Initialization and Warmstart}
\State $y = \phi_1(x)$
\State $z_{curr} = z$
\Repeat 
    \State $z_{next} = h(x, y, z_{curr})$
    \State $\Delta = \|z_{next} - z_{curr}\|$
    \State $z_{curr} = z_{next}$
\Until{$\Delta \leq \sigma$}


\State \textbf{Phase 2: Main Sensitivity Loop}
\State $\ell \gets 0$
\State $\tilde{z}^0 = z_{curr}$
\State $\tilde{s}^0 = s$
\Repeat
    \State $\tilde{z}^{\ell+1} = h(x, y, \tilde{z}^\ell)$
    \State $\tilde{s}^{\ell+1} = \mathrm{J}_1 h(x, y, \tilde{z}^{\ell+1}) + \mathrm{J}_3 h(x, y, \tilde{z}^{\ell+1})\tilde{s}^\ell$
    \State $\ell \gets \ell + 1$
\Until{$\max\left\{(\eta)^\ell, \sum_{j=0}^\ell \eta^{\ell-j} \|\tilde{z}^{j+1} - \tilde{z}^j\|\right\} \leq \sigma$}
\State \textbf{Output:} $\bar{y} = y, \bar{z} = \tilde{z}^\ell, \bar{s} = \tilde{s}^\ell$
\end{algorithmic}
\end{algorithm}

\subsection{Convergence Analysis of Algorithm 1}
 In this subsection, we prove that Algorithm 1 converges to a DSE (i.e., WDSE, SDSE)
 such that the convergent point satisfying:
 \begin{equation}\label{condition}
  0\in \mathcal{J} \hat{U}_X(x) +\mathcal{N}_X(x)
 \end{equation}
Here, \(\mathcal{J}\hat{U}_X\) denotes the conservative Jacobian of \(\hat{U}_X\) at \(x\), and \(\mathcal{N}_X(x)\) denotes the normal cone to $\Omega_{x,i}$ at \(x\). 


For ease of analysis, we rewrite the projected gradient descent step in Algorithm 1 as
\begin{align*}
 x^{k+1} &= x^k + \beta^k \left( \mathbb{P}_{\Omega_{x,i}} \left[ x^k + \alpha^k {\nabla}\hat{U}_X^k \right] - x^k \right)
 \end{align*}
 Taking $\beta^k=1$, we have
 \begin{equation}
 \begin{aligned}
    x^{k+1}&= \mathbb{P}_{\Omega_{x,i}} [ x^k + \alpha^k {\nabla}\hat{U}_X^k ]
    =\mathbb{P}_{\Omega_{x,i}}[x^k + \alpha^k \zeta^k]\\
    &+(\mathbb{P}_{\Omega_{x,i}} [ x^k + \alpha^k {\nabla}\hat{U}_X^k ]-\mathbb{P}_{\Omega_{x,i}} [x^k + \alpha^k \zeta^k])\\
    &=x^k+\alpha^k(\xi^k+e^k) \quad \ \hspace{4cm} 
 \end{aligned}
\end{equation}
 where $\zeta^k\in \underset{\hat{\zeta}\in\mathcal{J}\hat{U}_X(x^k)}{{\arg\min}}\ \|\mathbb{P}_{\Omega_{x,i}}[x^k + \alpha^k \hat{U}_X^k] - \mathbb{P}_{\Omega_{x,i}}[x^k + \alpha^k \hat{\zeta}] \|$,
 \begin{equation}\label{error}
 \begin{aligned}
 \xi^k :&= -\frac{1}{\alpha^k}(x^k - \mathbb{P}_{\Omega_{x,i}}[x^k + \alpha^k \zeta^k])\\
  e^k :&= \frac{1}{\alpha^k}(\mathbb{P}_{\Omega_{x,i}}[x^k + \alpha^k \hat{U}_X^k] - \mathbb{P}_{\Omega_{x,i}}[x^k + \alpha^k {\zeta}^k]) \quad 
 \end{aligned} 
\end{equation}
Next, we establish the convergence of Algorithm 1 using Theorem 3.2 in \cite{davis2018stochasticsubgradientmethodconverges}. To this end, it remains to verify the following conditions:
\begin{enumerate}
  \item The function $\hat{U}_X(x) $ is definable in $\Omega_{x,i}$.
  \item the sequence $\alpha^k e^k$ is summable.
\end{enumerate}

The following lemmas show that both conditions are satisfied. Their proofs are provided in Appendix~\ref{algorithmsection}
\begin{lemma}
\label{definable}Under Assumptions~\ref{a 2.1}-\ref{a 2.3},
  the function $\hat{U}_X(x) $ is definable in $\Omega_{x,i}$.
\end{lemma}

The second condition can be satisfied by designing an appropriate stepsize sequence \( \{\alpha^k\} \). Among \( \alpha^k \) and \( e^k \), the stepsize \( \alpha^k \) can be explicitly designed, whereas \( e^k \) is determined in practice by the error \( \sigma^k \).

\begin{lemma}
\label{sumable}
  Suppose that Assumptions~\ref{delta} and \ref{a 2.1}--\ref{a 2.3} hold. Let $\{\alpha^k\}_{k \in \mathbb{N}}$ be a nonnegative, nonsummable, and square-summable sequence and $\{\sigma^k\}_{k \in \mathbb{N}}$ satisfy $\sum_{k=0}^\infty \alpha^k \sigma^k < \infty$.
  If $\mathcal{P}(x^k)$ is a singleton for all but finitely many iterates $x^k$ generated by Algorithm 1,
  then $\{\alpha^k e^k\}_{k \in \mathbb{N}}$, with $e^k$ defined in \eqref{error}, is summable.
\end{lemma}

 Let us analyze the convergence of Algorithm 1, whose proof is provided in Appendix~\ref{algorithmsection}


\begin{theorem}
\label{convergence}
  Let Assumptions 2.1--2.3 and 4.1--4.2 hold. Suppose $\{\alpha^k, \sigma^k\}_{k\in\mathbb{N}}$ satisfy the conditions in Lemma 4.5, and $\mathcal{P}(x^k)$ is a singleton for all but finitely many iterates on any interval $\Omega_{x,i}$. Then Algorithm 1 converges to a WDSE or an SDSE strategy.
\end{theorem}


In practice, explicitly obtaining all the zeros of the function \( f_2(x, \theta) \) is very hard. Nevertheless, we can approximate the zeros of \( f_2 \) together with an associated error bound. Consequently, it is necessary to quantify the error in the computed DSE resulting from the estimation of the zeros of \( f_2 \). For WDSE,  it is impossible for its utility to exceed that of every other strategy by a uniform positive constant; in other words, there always exist strategies whose utilities are arbitrarily close to the utility achieved under the WDSE. 
 Thus, we focus only on the relationship between the estimation error and the WDSE.  For a fixed parameter $\theta$,  let the zeros of $f_2(x,\theta)$ be  $\{x_1,x_2,\ldots,x_{q_\theta}\}$. In the absence of analytical solutions for these zeros, we can instead determine closed intervals $I_{j,\theta}$ such that $x_j\in I_{j,\theta}$. 
 The following theorem establishes a quantitative relation between the estimation accuracy of these zeros and the approximation error of the computed WDSE and shows that Algorithm 1 converges to an $\epsilon$-WDSE.


\begin{figure}[t]

\centering
\includegraphics[width=1\columnwidth]{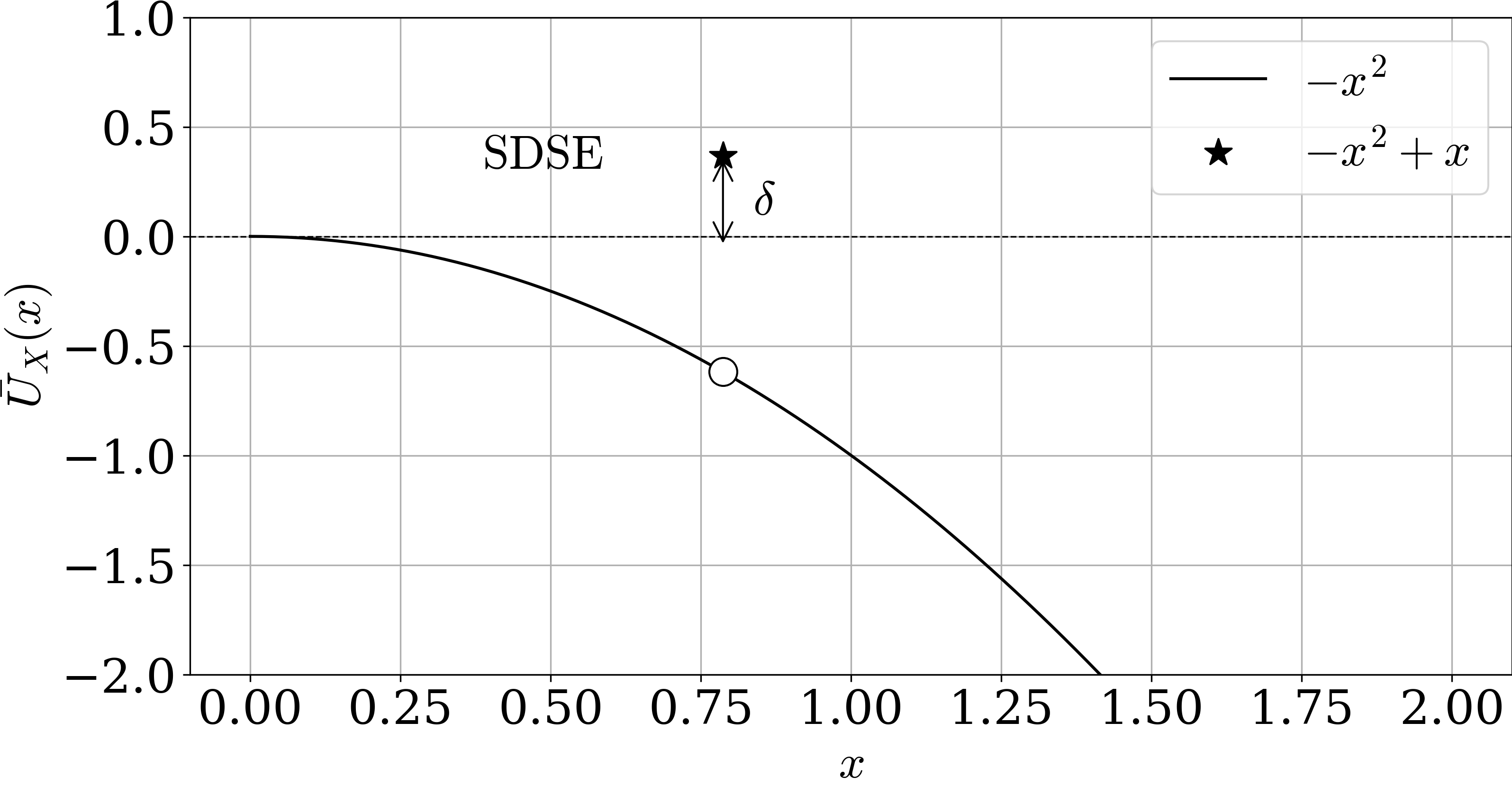}
\captionsetup{justification=centering}
\caption{The leader's utility.}
 \label{figexample 2}
\end{figure}

\begin{theorem}\label{theorem4.2}
  {
  Suppose Assumptions~\ref{concave}-\ref{delta} and \ref{a 2.1}-\ref{a 2.3} hold, where $\lambda(\cdot)$ denotes the Lebesgue measure on $\mathbb{R}$. Then there exists a constant $L$ such that  if $\lambda(I_{j,\theta}) \leq \frac{\epsilon}{L}$ for all $j$ and $\theta$, then Algorithm 1 converges to an $\epsilon$-WDSE strategy.
  }
\end{theorem}
The proof of Theorem~\ref{theorem4.2} is provided in Appendix~\ref{epsilonsection}

Unlike the WDSE case, for the SDSE, when the corresponding zero points do not admit a closed-form, the optimal value may be attained at analytically intractable zeros, in which case the associated equilibrium strategy cannot be explicitly recovered. Thus, we omit the analysis of SDSE here.
The following example illustrates this issue, where the deception parameters and the players $\boldsymbol{Z}$ are omitted for simplicity. Consider the following two-player setting,
\begin{equation}
\begin{aligned}
    &U_X=-x^2+(1-y)x\\
    &U_Y=(\sin x-\frac{x}{2})^2y+h(x)\\
    &x\in[0,1], y\in[-1,1], \Theta=\emptyset
\end{aligned}
\end{equation}
Under this setting, we obtain
\begin{equation} \label{error SDSE}
    \begin{aligned}
    \tilde{U}_X(x,\theta)&=\max_{y\in \text{BR}_Y(x,\theta)} \max_{\boldsymbol{z}\in \text{BR}_{\boldsymbol{Z}}(x,y,\theta)} U_X(x,y,\boldsymbol{z})\\
    &=\begin{cases}
        -x^2, \quad &2\sin x\neq x\\
        -x^2+x, \quad & 2 \sin x =x
    \end{cases}
    \end{aligned}
\end{equation}
As shown in  Fig.~\ref{figexample 2}, the leader’s optimal utility is attained at the solution of \( 2\sin x = x \), but this solution lacks a closed-form expression, and any deviation from the corresponding strategy yields a utility at least $\delta$ lower.

\section{Application scenarios}
\label{simulation}

In this section, we demonstrate our theoretical results 
with the two application scenarios, namely, secure wireless communication and the defense against IA-FDI attacks in microgrids.
\subsection{Secure Wireless Communication}

\begin{figure}[htbp]
    \centering
    \includegraphics[width=1\linewidth]{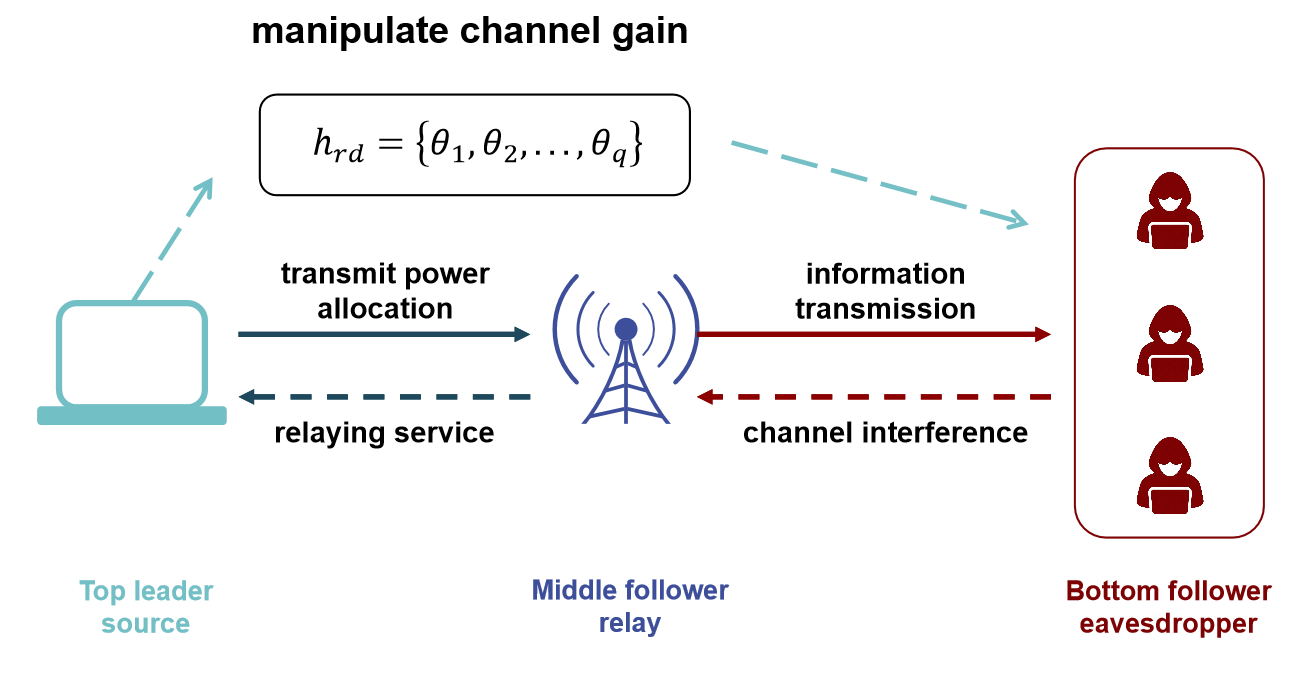}
    \caption{Wireless communication scenario}
    \label{fig:ppt1}
\end{figure}

In wireless communication, as
shown in Fig.~\ref{fig:ppt1}, the top-level leader $X$ is a source node aiming to maximize its secure transmission rate to a legitimate destination. To achieve this, the source purchases transmit power $x$ from a relay (middle-level follower $Y$), which sets a unit price $y$ to maximize its profit. Simultaneously, a set of malicious eavesdroppers (bottom-level followers $\boldsymbol{Z}$) choose an interference power $\boldsymbol{z}$ to disrupt the communication link.  The source utilizes its private knowledge of the true channel state information to broadcast a strategically manipulated signal. {Specifically, the source employs a deception strategy by misrepresenting the channel quality, effectively distorting the eavesdroppers' perception of the propagation environment},  thereby misleading the eavesdroppers into adopting suboptimal jamming strategies.
The true Signal-to-Interference-plus-Noise Ratio (SINR) at the destination is defined as $\mathrm{SINR}_{d_0}(x) = \frac{|h_{{rd}_0}|^2 x}{\eta+ |h_{ed}|^2(\sum\limits_{i=1}^N z_i) }$, while the SINR perceived by the eavesdroppers under the deceptive signal $h_{rd}$ is $\mathrm{SINR}_d(x) =\frac{|h_{{rd}}|^2 x}{\eta+ |h_{ed}|^2\sum\limits_{i=1}^N z_i}$.
Then the game for the three parties is modeled by the following optimization problems \cite{8454822,8017519}
\begin{equation}
\begin{aligned}
&\max_{x \in \Omega_x} U_X(x, y, \boldsymbol{z}, \theta_0) = d_1 \mathrm{SINR}_{d_0}(x) - d_4 x y \\
&\max_{y \in \Omega_y} U_Y(x, y, \boldsymbol{z},\theta) = x y - d_3 x \\
&\max_{z_i \in \Omega_{z_i}} U_{Z_i}(x, y, z_i, \boldsymbol{z}_{-i},\theta) = -\log_2 \big(1 + \mathrm{SINR}_d(x)\big) - d_{2i} z_i
\end{aligned}
\end{equation}
where $d_1$ represents the gain coefficient, and $d_{2i}, d_3, d_4$ are cost coefficients.

The Blind setting shown in purple in  Fig.\ref{fig: impact of deception} represents the worst-case utility for the leader when it is unable to determine whether the follower possesses the capability to observe its actions. The deception parameter $\theta\in[0.5,1]$ is defined as the manipulated channel quality $h_{rd}$, and the optimal deception parameter is attained at $\theta=0.5$. As depicted in Fig.~\ref{fig: impact of deception}, the deception parameter $\theta$ serves as a control knob: in particular, choosing $\theta=0.5$ not only yields the optimal leader utility in the leader--follower scheme, but also ensures robustness against uncertainty in the follower's decision scheme.

\begin{figure}[htbp]
    \centering
    \includegraphics[width=0.9\linewidth]{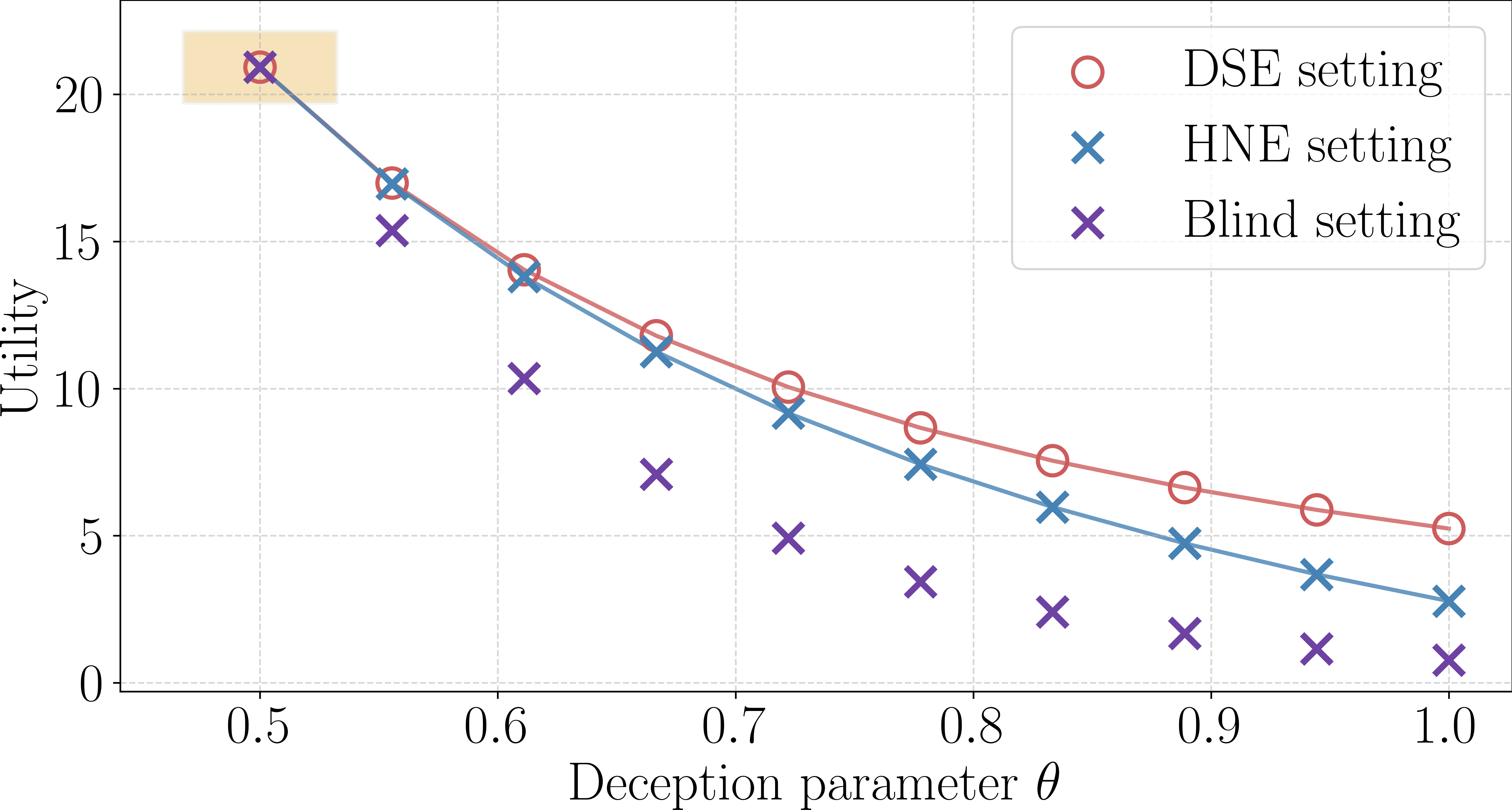}
    \caption{The Impact of Deception}
    \label{fig: impact of deception}
\end{figure}

In Fig.~\ref{fig:convergence_comparison}, we evaluate the convergence performance of  Algorithm 1 in seeking the WDSE. 
The system parameters are set as follows.  Set $\Theta=\{0.8,1,1.2\}$. The number of eavesdroppers is $N=3$. The channel power gains are set to $\theta_0=h_{rd_0} = 1.0$ and $h_{ed} = 0.8$. The background noise is $\eta = 0.5$. Regarding the utility coefficients, we set the gain parameter $d_1 = 10$, and the cost parameters $d_3=2$, $d_4 = 7$ and $\mathbf{d}_2 = [0.5, 0.6, 0.7]^\top$. The strategy space for the insider is bounded by $y \in [0, 2]$, and step size $\alpha^k=\frac{1}{k^{0.6}}$ and tolerance $\sigma^k=\frac{1}{k^{1.1}+1}$. In this setting, \( \mathrm{BR}_y(x, \theta) \) and \( \mathrm{BR}_{\boldsymbol{z}}(x, y, \theta) \) are single-valued and thus, the WDSE and SDSE are equivalent.


As shown in Fig.~\ref{fig:convergence_comparison}, the solid lines represent the iterative updates generated by Alg.~1, while the dashed lines denote the WDSE.  
{In this setting, the optimal deception parameter is $\theta^*=0.8$}.
The results demonstrate that both the strategy sequence and the corresponding utility values converge to the  WDSE. 
\begin{figure}[htbp]           
    \centering                 
    \includegraphics[width=0.51\textwidth]{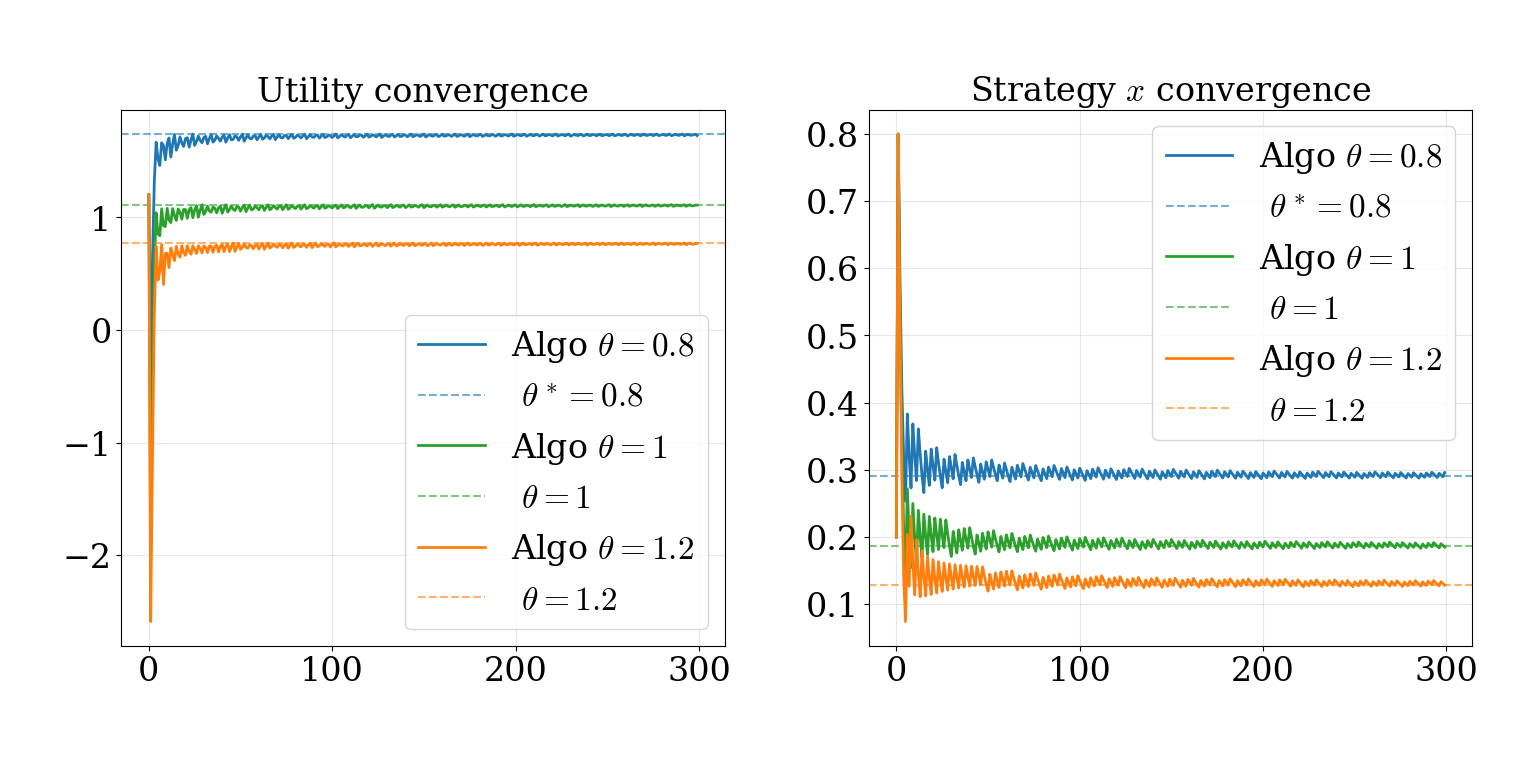}
    \caption{Convergence performance of Alg.~1}   
    \label{fig:convergence_comparison}        
\end{figure}

Next, we demonstrate the consistency between WDSE and HNE.  The parameters
$d_1,d_3,d_4$, serving as gain and cost coefficients, significantly influence secure wireless communication. Therefore, we focus on $d_1,d_3$, and $d_4$ and examine how their variation affects consistency.
The simulation results are presented in Fig.~\ref{consistencefig}.

\begin{figure}[htbp]           
    \centering                 
 \includegraphics[width=0.5\textwidth]{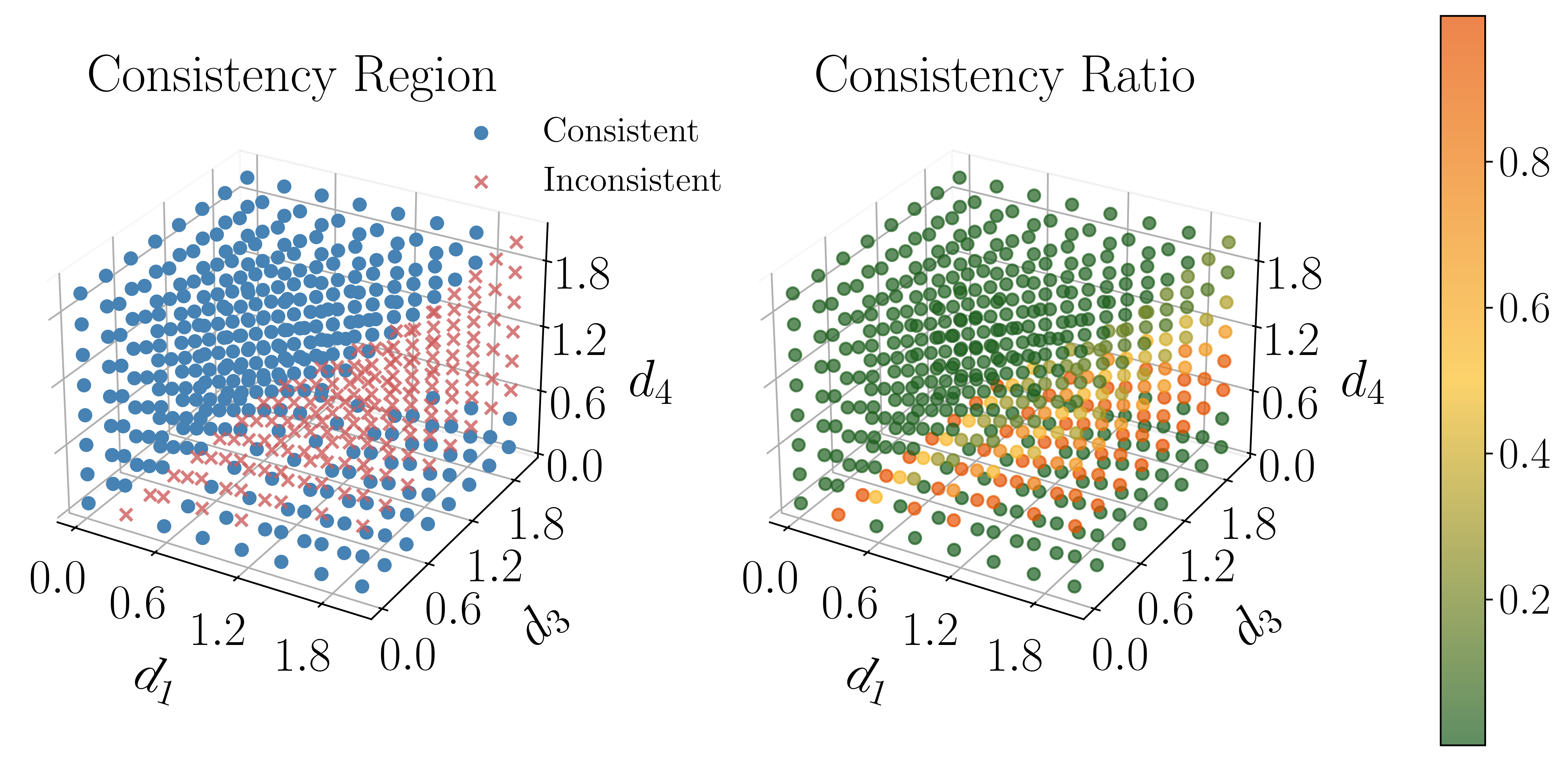} 
    \caption{The relationship between WDSE and HNE under different environment settings.}   
    \label{consistencefig}        
\end{figure}

 As observed in Fig.~\ref{consistencefig}, when the ratio $\frac{d_4}{d_1}$ is relatively small or large, the WDSE is consistent with the HNE. Consequently, the source can confidently adopt the WDSE strategy. The right panel of Fig.~\ref{consistencefig} presents the error heatmap illustrating the deviation between WDSE and HNE under various parameter settings. As observed, the blue regions denote near-perfect consistency, indicating that WDSE aligns closely with HNE. In contrast, the cyan-green regions signify a significant deviation between the two equilibria. The results indicate that, as long as the consistency condition holds, the channel environment exhibits robustness, relieving the sender of any strategic selection dilemma.

\subsection{Defense against IA-FDI attack in microgrid}

\begin{figure}[htbp]
    \centering
    \includegraphics[width=1\linewidth]{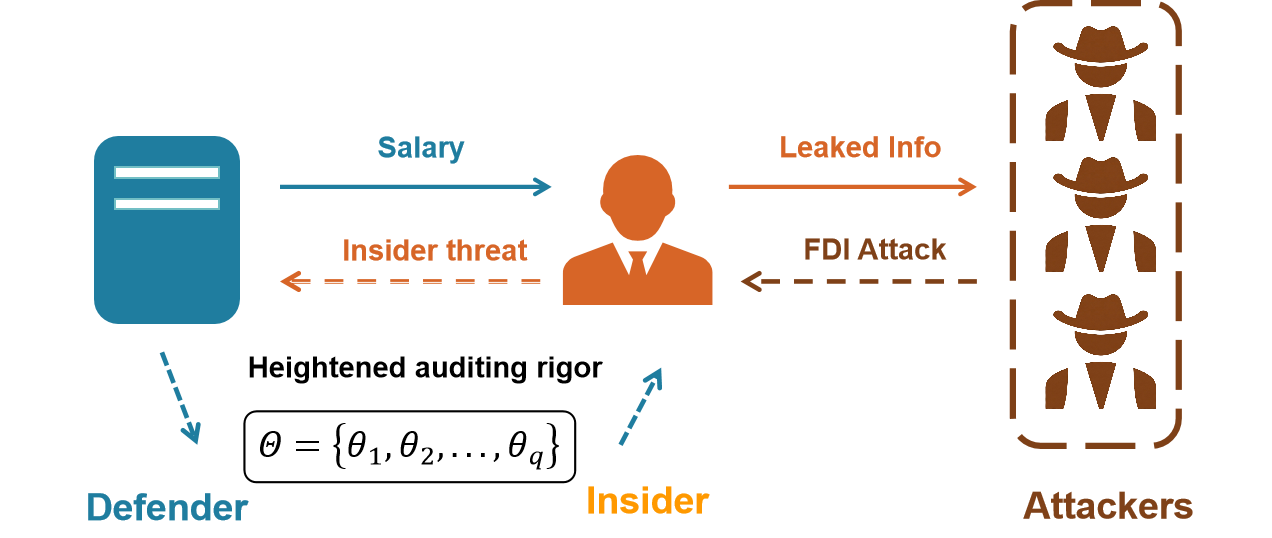}
    \caption{Defense against IA-FDI attack in microgrid}
    \label{fig:ppt2}
\end{figure}

Consider a system defender aiming to safeguard the power grid against false data injection attacks while minimizing operational and incentive costs shown in Fig.~\ref{fig:ppt2}.  The defender  $X$ determines a salary $x$ as an incentive for the insider to protect the system. {Simultaneously, the defender strategically leverages private information regarding the penalty mechanism to signal a heightened level of auditing rigor, thereby inflating the insider's perceived cost of betrayal.}
The insider $Y$ observes this monetary signal and weighs the risk of betrayal against potential external bribes, thereby determining the probability $y$ of leaking internal information, such as critical topological data of the target power system.
Consequently, the external attackers $\boldsymbol{Z}$ adjust their false data injection intensity $\boldsymbol{z}$ in response to the insider’s leaked information, injecting false voltage or current data into the monitoring system to maximize the damage to the power grid. 
The deception refers to the defender's private information about the penalty parameter $\beta$.

The expected damage to the power grid is defined as $D(y, \boldsymbol{z}) = \lambda (1 + \alpha y) \sum_{i=1}^N z_i$, where the information leakage $y$ from the insider amplifies the impact of the attack vectors $\boldsymbol{z}$ injected by attackers.
The optimization problems for the three parties are modeled as
\begin{equation}
\begin{aligned}
&\max_{x \in \Omega_x} U_X(x, y, \boldsymbol{z}) =L - x - \lambda b(y) \sum_{i=1}^N z_i \\
&\max_{y \in \Omega_y} U_Y(x, y) = (V_{\text{base}} - \beta x) y + x \\
&{\max_{z_i \in \Omega_{z_i}} U_{Z_i}(x, y, z_i, \boldsymbol{z}_{-i}) = b(y) z_i - \frac{1}{2} c(x) z_i^2 + \sum_{j\neq i}g_{ij}z_iz_j}
\end{aligned}
\end{equation}
Here, $L$ denotes the initial system utility or the intrinsic value of the microgrid assets under normal operation, $\lambda$ denotes the system loss coefficient per unit of attack intensity, and $\alpha$ captures the amplification effect of insider information leakage. $V_{\text{base}}$ denotes the insider’s baseline bribery utility, while $\beta$ quantifies the penalty imposed by the defender. For attackers, $b(y)=1+\alpha y$ and $c(x)=1+\delta x$ represent the marginal benefit and cost coefficients, respectively, and $g_{ij}$ characterizes the coupling strength between attackers.

The default parameters in this setting are configured as follows. The system consists of $N=3$ attackers. The loss coefficient and amplification factor are set to $\lambda = 1.5$ and $\alpha = 0.5$, respectively. For the insider, the baseline utility is $V_{\text{base}} = 3.0$ with a penalty factor $\beta = 1.0$. The cost scaling parameter for attackers is $\delta = 0.8$. The interaction matrix $G$ is set to be uniform with coupling strength $g_{ij} = 0.2$ for all $i \neq j$ and $0$ on the diagonal. The strategy spaces are bounded by $x \in [0, 5]$, $y \in [0, 1]$, and $z_i \in [0, 5]$. Set $\Theta=\{\beta\}$, implying a fixed deception environment.

The absence of insider modeling leads to a marked degradation in leader utility. As demonstrated in Fig.~\ref{fig:insider},  the left $y$-axis reports the leader's utility, while the right $y$-axis shows the utility gap between the insider-aware and insider-unaware cases. When the defender adopts a fixed salary policy after implementing a defense strategy without accounting for insider behavior, the defender's utility drops significantly compared to the insider-aware model.
\begin{figure}[htbp]
    \centering
    \includegraphics[width=1\linewidth]{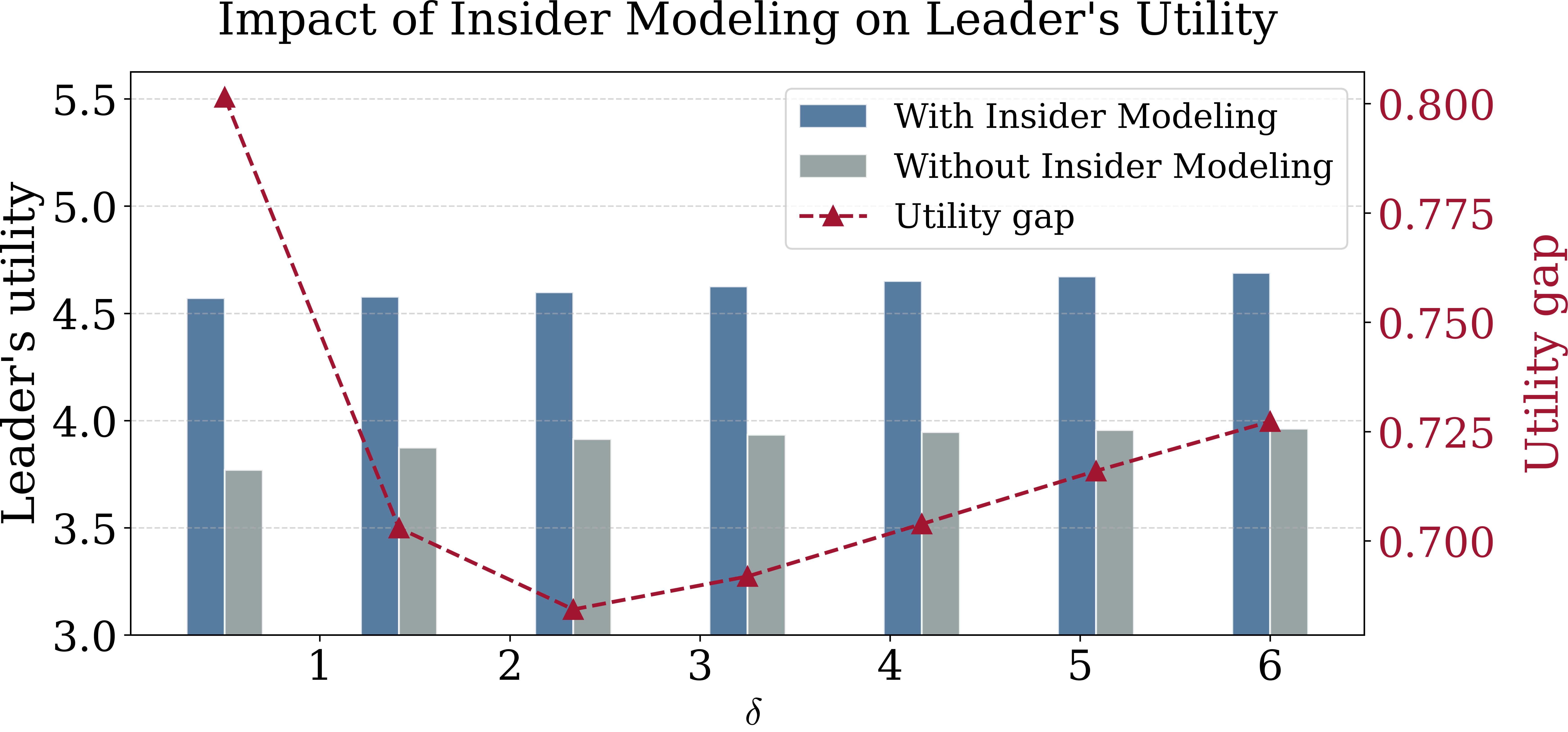}
    \caption{The  role of insider}
    \label{fig:insider}
\end{figure}

We next show the convergence to an $\epsilon$-WDSE even when the zeros of the follower's reaction function $f_2(\cdot, \theta)$ cannot be solved explicitly. 
In this experiment, we partition the leader's strategy space $\Omega_x=[0,5]$ into sub-intervals. Let $\tilde{\Omega}_{x,1}=[0,2.8]$, $\tilde{\Omega}_{x,2}=[3.1,5]$, and the gap interval $I_1=[2.8,3.1]$.

As shown in Fig.~\ref{epsilon-WDSEfig}, even when the zero point can only be localized within the interval $I_1$, an $\varepsilon$-WDSE is still attainable. The leader may then adopt this solution as an approximation to the optimal strategy, ensuring that its utility deviates from the supremum of achievable utilities by at most $\epsilon$.  Fig.~\ref{interval bound} shows that as the length of the gap interval $I_1$ containing the zero decreases, the gap between the utility achieved by our algorithm and the optimal utility gradually diminishes.

\begin{figure}[htbp]
    \centering       \hspace{-0.5cm}   \includegraphics[width=0.517\textwidth]{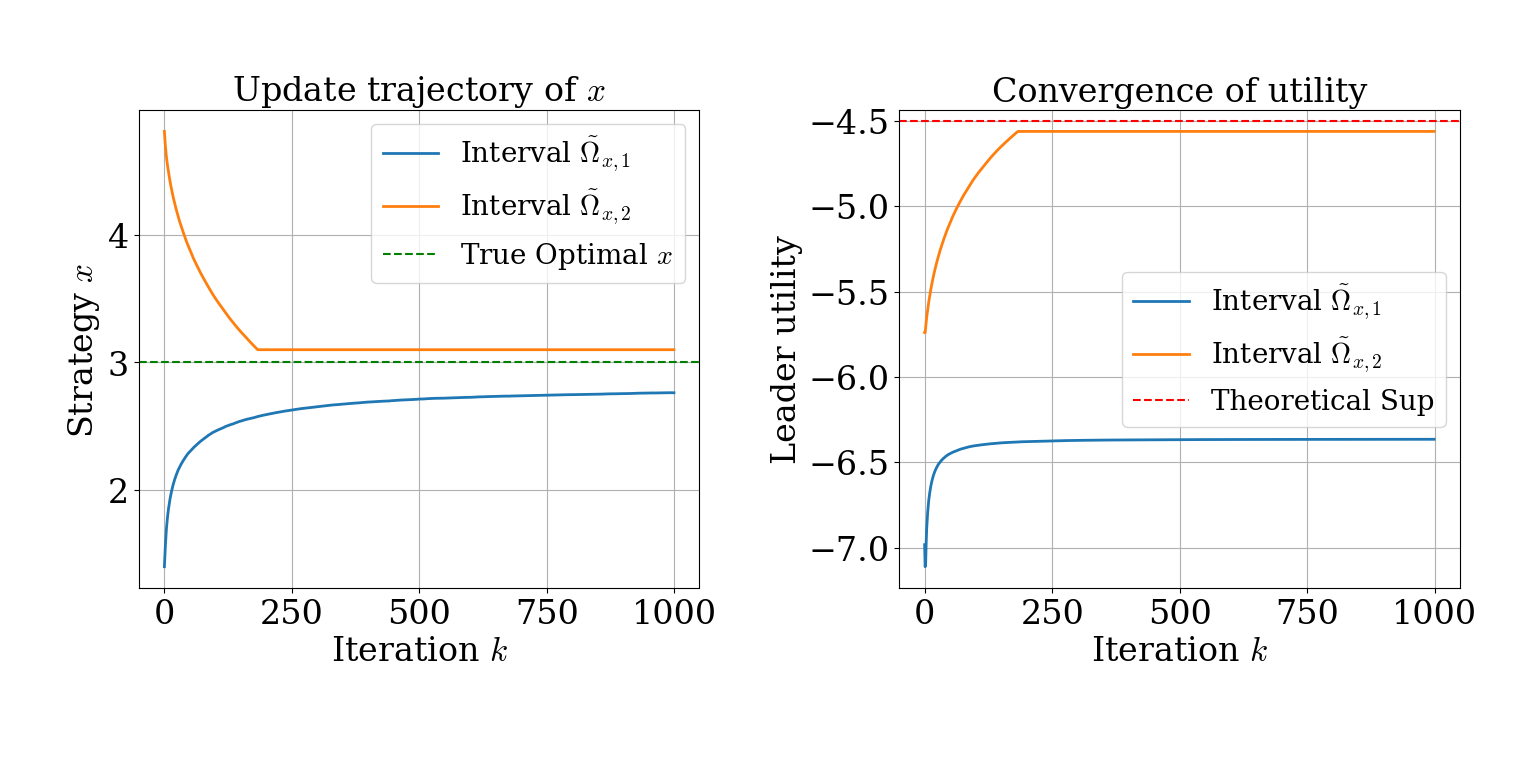}
    \caption{ $\epsilon$-WDSE convergence}   
    \label{epsilon-WDSEfig}  
\end{figure}

\begin{figure}[htbp]
    \centering
    \includegraphics[width=0.9\linewidth]{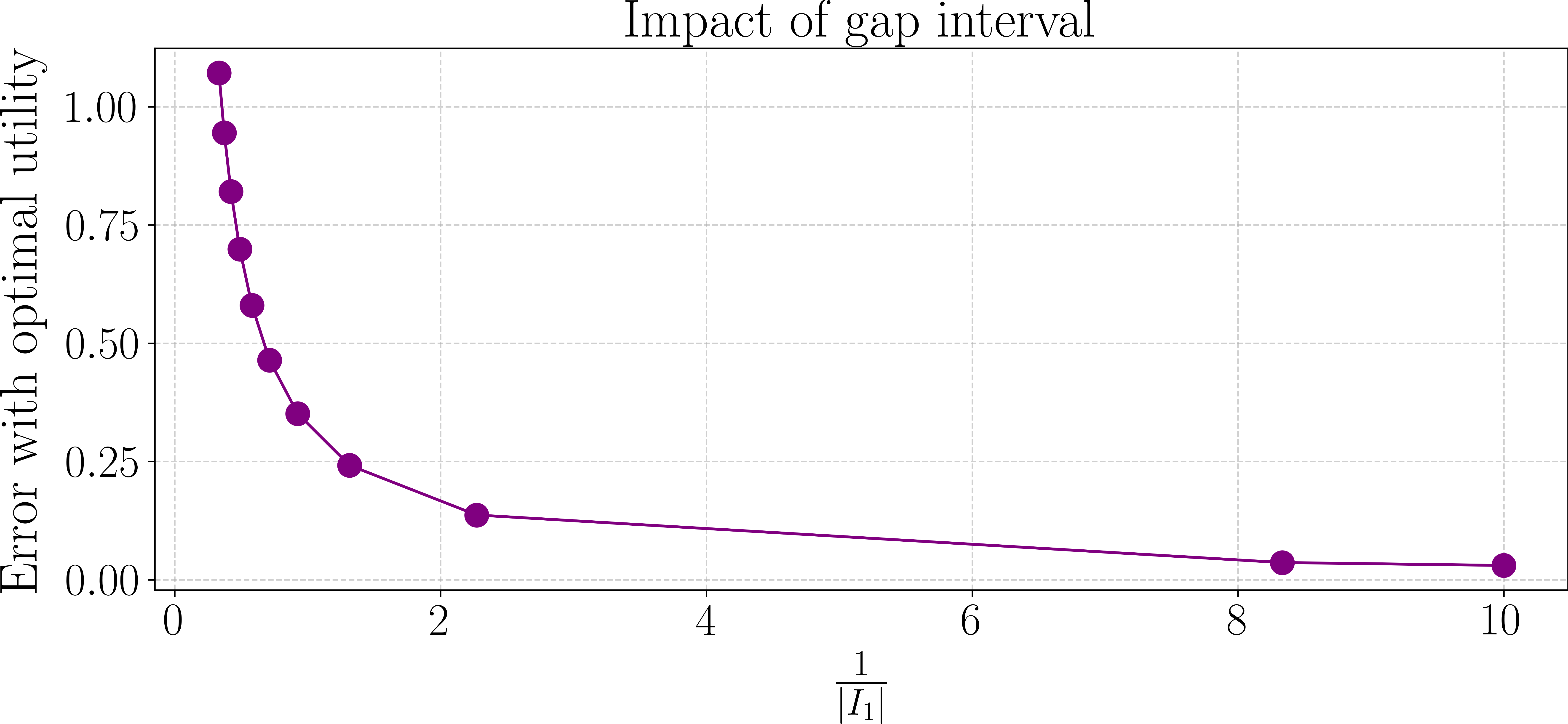}
    \caption{The relation between $\epsilon$ and gap interval $I_1$}
    \label{interval bound}
\end{figure}

Setting $\Theta=\{0.8,1,1.2\}$, we obtain $\theta^*=0.8$. 
Fig.~\ref{fig:consistency_utility} illustrates the Leader's utility under  $(x^*, y^*, z^*,\theta^*)$, $(x^*, y^\diamond, z^\diamond,\theta^*)$, and $(x^\diamond, y^\diamond, z^\diamond,\theta^*)$. These evaluations are conducted across a set of varying parameter values $\delta \in \{0.8, 1.4, 2.0, 2.6, 3.0\}$.
As shown in Fig.~\ref{fig:consistency_utility}, when $\delta\in [0.8,1.4]$, the leader’s utility does not decrease regardless of whether insiders and attackers adopt the WDSE or HNE strategy. Hence, under this parameter regime, the defender can guarantee robustness of its utility by employing the DSE strategy.
\begin{figure}[htbp]
    \centering
    \includegraphics[width=0.9\linewidth]{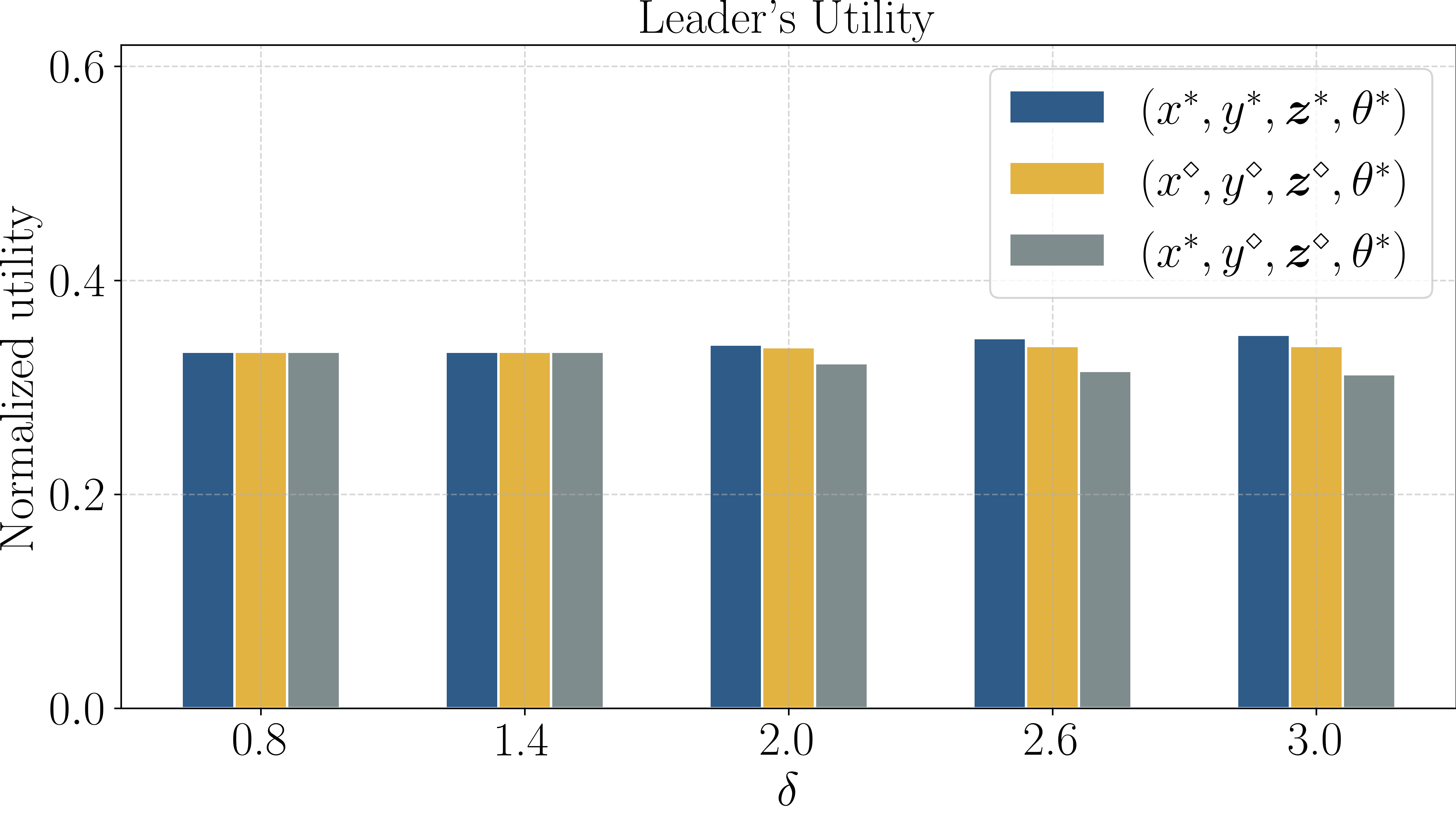}
    \caption{Insider takes DSE strategies v.s. HNE strategies.}
    \label{fig:consistency_utility}
\end{figure}

As the number of followers in our model increases, the convergence time of our algorithm does not grow exponentially. The coupling matrix $G=[g_{ij}]_{i,j=1}^n$ is randomly generated to characterize the mutual interference among attackers. Tab. 1 shows the convergence time of our algorithm under varying numbers of agents.

\begin{table}[htbp]
\centering
\caption{Scalability and Computation Time Analysis}
\begin{tabular}{@{}lccccc@{}}
\toprule
\textbf{Number of followers ($N$)} & 50 & 100 & 150 & 200 & 250 \\ \midrule
\textbf{Execution Time (sec)}      & 0.1418 & 0.1804 & 0.3414 & 1.3611 & 1.8228 \\ \bottomrule
\end{tabular}
\end{table}

\section{Conclusion}
This paper investigated a three-party game involving an insider, where the leader maximizes its utility through active deception. We established a unified framework to analyze the DSE and derived necessary and sufficient conditions for its consistency with the HNE. This analysis provides a theoretical basis for designing robust deception signals. To address the computational challenges of non-smooth and set-valued BR mappings, we proposed a scalable hyper-gradient-based algorithm. This method guarantees convergence to a WDSE or SDSE, and relaxes to an $\epsilon$-WDSE
when exact BR mappings are unattainable or a WDSE does not exist. Furthermore, we validated our framework in practical scenarios, including secure wireless communication and defense against insider-assisted false data injection attacks. 

Future research will focus on extending both the theoretical analysis and the algorithmic framework to broader settings, particularly considering cases where players face polyhedral strategy constraints, and the deception parameter set $\Theta$ is a compact convex set.
\appendices
\section{}\label{consistencesection}
\noindent\textbf{Proof of Theorem~\ref{consistence}}

\textbf{Sufficiency:}
Let $(x^*,y^*,\boldsymbol{z}^*,\theta^*)$ be a WDSE. If $T_1(x^*)=0$, then due to the concavity of $U_X$ in $x$, $x^*$ is a global maximizer of $U_X(x, y^*, \boldsymbol{z}^*, \theta_0)$. Thus, $x^*=x^{\diamond}$, and  consistency holds.
In the following, we consider \( T_1( x^* ) \neq 0 \).

Consider the case where Condition 3 holds, namely that $\tilde{U}_X(x,\theta^*)$ is differentiable at $x^*$. Since $\tilde{U}_X$ is piecewise differentiable, there exists a neighborhood $\delta(x^*)$  in which the derivative exists.
Assume $T_2(x^*,\theta^*)>0$. Since the condition requires $T_1(x)T_2(x,\theta^*)>0$, we must have $T_1(x^*)>0$.
If \( x^* \) is an interior point of \( \Omega_x \), the first-order necessary condition for WDSE would imply $T_2(x^*,\theta^*)=0$, contradicting the assumption that $T_2 > 0$.
If $x^*=x_{\min}$, given $T_2(x^*,\theta^*)>0$, there would exist a point $x > x^*$ such that $\tilde{U}_X(x,\theta^*) > \tilde{U}_X(x^*,\theta^*)$. This contradicts the definition of $x^*$ as the WDSE strategy.
Therefore, we must have $x^*=x_{\max}$. With $x^*=x_{\max}$ and $T_1(x^*)>0$, combined with the concavity of $U_X$, it follows that $U_X$ is increasing on $\Omega_x$ and attains its maximum at the boundary. Thus,
\begin{equation}
    x^* \in \underset{x \in \Omega_x}{\arg \max} \ U_X(x, y^*,\boldsymbol{z}^*, \theta_0),
\end{equation}
which implies $x^* = x^\diamond$. The proof for the case \( T_2(x^*, \theta^*) < 0 \) is analogous.


Consider the case where Condition 2 holds, in which  $\tilde{U}_X(x,\theta^*)$ is not differentiable at $x^*$. There exists a punctured neighborhood $\mathring{\delta}(x^*)$ where differentiability holds. Assume $T_1(x)>0$ for all $x\in\mathring{\delta}(x^*)$, which implies $T_2(x,\theta^*)>0$ for all $x\in\mathring{\delta}(x^*)$. If $x^*\neq x_{\max}$, there exists a point $x_0 > x^*$ within the neighborhood. By the Lagrange Mean Value Theorem, there exists $\xi \in (x^*, x_0)$ such that
\begin{equation}
    \tilde{U}_X(x_0,\theta^*)-\tilde{U}_X(x^*,\theta^*)=T_2(\xi,\theta^*)(x_0-x^*)>0.
\end{equation}
This implies $\tilde{U}_X(x_0,\theta^*) > \tilde{U}_X(x^*,\theta^*)$, contradicting the optimality of $x^*$. Thus, $x^*=x_{\max}$. Since $T_1(x) > 0$ near $x_{\max}$, $x^*$ maximizes $U_X$, $x^*=x^{\diamond}$. The proof for the case \( T_1(x) < 0 \) is analogous.

\textbf{Necessity:}
Let $(x^*,y^*,\boldsymbol{z}^*,\theta^*)$ be both a WDSE and an HNE (i.e., $x^*=x^{\diamond}$).

If $x^*$ is an interior point of $\Omega_x$, then $T_1(x^*)=0$ is required for $(x^*,y^*,\boldsymbol{z}^*,\theta^*)$ to be an HNE, satisfying Condition 1.
Next, consider the boundary case \( x^* = x_{\max} \) (the case \( x^* = x_{\min} \) is analogous).
Since $x^*=x_{\max}$ maximizes $U_X$, $T_1(x^*) \geq 0$. If $T_1(x^*)=0$, Condition 1 is met.
If $T_1(x^*) > 0$, by the continuity of the derivative, there exists a neighborhood $\delta_1(x^*)$ such that $T_1(x) > 0$ for all $x \in \delta_1(x^*) $.

 Moreover, by the definition of the WDSE, there exists a neighborhood $\mathring\delta_2(x^*)$ such that $\tilde{U}_X(x,\theta^*)<\tilde{U}_X(x^*,\theta^*)$. Because $\tilde{U}_X(x,\theta^*)$ is definable,  $T_2(x,\theta^*)$ is definable. Thus, there exists a punctured neighborhood $\mathring{\delta}_3(x^*)$ satisfying $T_2(x,\theta^*)>0$ and let $\delta_4(x^*)=\cap_{i=1}^3\mathring{\delta}_i(x^*)$. In this neighborhood, $T_1(x)T_2(x, \theta^*) > 0$, satisfying Condition 2 or 3 based on differentiability.  \hfill$\square$

\section{}\label{algorithmsection}
\noindent\textbf{Proof of Lemma~\ref{contraction} }

With $h(x,y,\boldsymbol{z})=\mathbb{P}_{\Omega_z}[\boldsymbol{z}-\gamma F(x,y,\boldsymbol{z})]$, 
\begin{equation}
    \begin{aligned}
        &\|h(x,y,\boldsymbol{z}_1)- h(x,y,\boldsymbol{z}_2)\|^2\\
        &\leq \|(\boldsymbol{z}_1-\gamma F(x,y,\boldsymbol{z}_1))-
        (\boldsymbol{z}_2-\gamma F(x,y,\boldsymbol{z}_2))\|^2 \\
        &=\|\boldsymbol{z}_1-\boldsymbol{z}_2\|^2+
        \gamma^2\|F(x,y,\boldsymbol{z}_1)-F(x,y,\boldsymbol{z}_2)\|^2 \\
        &-2\gamma(\boldsymbol{z}_1-\boldsymbol{z}_2)^\top(F(x,y,\boldsymbol{z}_1)-F(x,y,\boldsymbol{z}_2))
         \\
        &\leq (1-2\gamma \mu +\gamma^2 \kappa^2) \|\boldsymbol{z}_1-\boldsymbol{z}_2\|^2 \\
    \end{aligned}
\end{equation}

    Therefore, for any $\gamma<\frac{2\mu}{\kappa^2}$, $\eta=\sqrt{1-2\gamma \mu +\gamma^2 \kappa^2}<1$.
    Thus, $h(x,y,\boldsymbol{z})$ is a contraction mapping about $\boldsymbol{z}$, and the sequence generated by \eqref{phi_2} converges to the unique fixed point of $h(x,y,\boldsymbol{z})$ linearly with rate $\eta$. \hfill$\square$

\noindent \textbf{Proof of Lemma~\ref{offline_lemma}}
 
Since $h(x,\phi_1(x),\cdot)$ is differentiable at $\phi_2(x)$, and $h(x,\phi_1(x),\cdot)$ is a contraction mapping
 with contraction constant $\eta$, 
$\|\mathrm{J}_3h(x,\phi_1(x),\phi_2(x))\|\leq \eta<1.$
 Notice that $\phi_2(x)$ is the unique
 solution of the equation $\boldsymbol{z}=h(x,\phi_1(x),\boldsymbol{z})$. By the implicit function theorem, $\phi_2(x)$
is continuously differentiable at $x$ and
\begin{equation}
     \mathrm{J}\phi_2(x)=
\left(I-\mathrm{J}_3h(x,\phi_1(x),\phi_2(x))\right)^{-1}\mathrm{J}_1h(x,\phi_1(x),\phi_2(x)).
\end{equation}
Clearly, $\|\mathrm{J}_3h(x,\phi_1(x),\phi_2(x))\|<1$  implies that \eqref{fixed_point} is contractive and converges to the unique fixed point $\mathrm{J}\phi_2(x)$. \hfill$\square$

\noindent \textbf{Proof of Lemma~\ref{online_lemma}}
 
Since $|\mathcal{P}(x)|=1$,  $\omega(x,\phi_1(x),\phi_2(x))\in \text{int}\, A_i$.
Thus, there exists $\epsilon>0$ such that $B(\omega(x,\phi_1(x),\phi_2(x)),\epsilon)\subset \text{int}\, A_i$.
Hence, by continuity of $\omega(x,\phi_1(x),\cdot)$, there exists $\delta>0$ such that $\boldsymbol{z}\in B(\phi_2(x),\delta)$, which implies $\omega(x,\phi_1(x),\boldsymbol{z})\in B(\omega(x,\phi_1(x),\phi_2(x)),\epsilon)\subset \text{int}\, A_i$.
Since the sequence $\{\tilde{\boldsymbol{z}}^\ell\}_{\ell=0}^\infty$ converges to $\phi_2(x)$, there exists $\bar{\ell}\in \mathbb{N}$ such that for all $\ell\geq \bar{\ell}$,  $\tilde{\boldsymbol{z}}^\ell\in B(\phi_2(x),\delta)$, which implies $\omega(x,\phi_1(x),\tilde{\boldsymbol{z}}^\ell)\in \text{int} A_i$. Then for $\ell\geq \bar{\ell}$, 
\begin{equation}
\begin{aligned}
    &\|\mathrm{J}_1h(x, \phi_1(x), \tilde{\boldsymbol{z}}^\ell) - \mathrm{J}_1h(x, \phi_1(x), \phi_2(x))\|\\
    \leq &\gamma\|\mathrm{J}{\mathbb{P}_{\Omega_z}}[\omega(x,\phi_1(x),\phi_2(x))]\|\ \|\mathrm{J}_1 F(x,\phi_1(x),\tilde{\boldsymbol{z}}^\ell)\\
    -&\mathrm{J}_1 F(x,\phi_1(x),\phi_2(x))\|\\
    \leq &\gamma L_{F}\|\tilde{\boldsymbol{z}}^\ell - \phi_2(x)\|.
\end{aligned}
\end{equation}
    We also have
    \begin{equation}
\begin{aligned}
    &\|\mathrm{J}_3h(x, \phi_1(x), \tilde{\boldsymbol{z}}^\ell) - \mathrm{J}_3h(x, \phi_1(x), \phi_2(x))\|\\
    =&\|\mathrm{J}{\mathbb{P}_{\Omega_z}}[\omega(x,\phi_1(x),\tilde{\boldsymbol{z}}^\ell)](I-\gamma\mathrm{J}_3 F(x,\phi_1(x),\tilde{\boldsymbol{z}}^\ell))\\
    &-\mathrm{J}{\mathbb{P}_{\Omega_z}}[\omega(x,\phi_1(x),\phi_2(x))](I-\gamma\mathrm{J}_3 F(x,\phi_1(x),\phi_2(x)))\|\\
    \leq & \gamma L_{F}\|\tilde{\boldsymbol{z}}^\ell - \phi_2(x)\|   .
\end{aligned}
\end{equation}

By the triangle inequality,
\begin{equation}
    \|\tilde{\boldsymbol{s}}^{\ell+1}-\mathrm{J}\phi_2(x)\|\leq \|\tilde{\boldsymbol{s}}^{\ell+1}-\hat{s}^{\ell+1}\|+\|\hat{s}^{\ell+1}-\mathrm{J}\phi_2(x)\|
\end{equation}

 For the sake of brevity, take $L^{\ell}=\mathrm{J}_1 h(x,\phi_1(x),\tilde{\boldsymbol{z}}^{\ell})$,
$L^{*}=\mathrm{J}_1 h(x,\phi_1(x),\phi_2(x))$,
$M^{\ell}=\mathrm{J}_3h(x,\phi_1(x),\tilde{\boldsymbol{z}}^{\ell})$ and
$M^{*}=\mathrm{J}_3 h(x,\phi_1(x),\phi_2(x))$. 
For the first term, we have
\begin{equation}\label{off-on_1}
\begin{aligned}
&\Vert \tilde{s}^{\ell+1} - \hat{s}^{\ell+1} \Vert = \Vert M^{\ell}\tilde{s}^{\ell} + L^{\ell} - M^{\ast}\hat{s}^{\ell} - L^{\ast} \Vert \\
&{\leq} \Vert M^{\ell} \Vert \Vert \tilde{s}^{\ell} - \hat{s}^{\ell} \Vert + \Vert M^{\ell} - M^{\ast} \Vert \Vert \hat{s}^{\ell} \Vert + \Vert L^{\ell} - L^{\ast} \Vert \\
&{\leq} \eta \Vert \tilde{s}^{\ell} - \hat{s}^{\ell} \Vert + \gamma L_{F} \Vert \tilde{\boldsymbol{z}}^{\ell} - \phi_2(x) \Vert \Vert \hat{s}^{\ell} \Vert 
 + \gamma L_{F} \Vert \tilde{\boldsymbol{z}}^{\ell} - \phi_2(x) \Vert \\
&{\leq} \frac{\gamma L_{F} \Vert \hat{s}^{\ell} \Vert + \gamma L_{F}}{1 - \eta} \Vert \tilde{\boldsymbol{z}}^{\ell+1} - \tilde{\boldsymbol{z}}^{\ell} \Vert + \eta \Vert \tilde{s}^{\ell} - \hat{s}^{\ell} \Vert \quad 
\end{aligned}
\end{equation}
Since $\|\hat{s}^\ell\|$ is generated by a contraction mapping \eqref{fixed_point}, there exists a constant $B>0$ such that $\|\hat{s}^\ell\|\leq B$ for all $\ell\in \mathbb{N}$.
Take $B_{ps}=\frac{\gamma L_{F} B + \gamma L_{F}}{1-\eta}$.
Then
\begin{equation}
    \|\tilde{s}^{\ell+1}-\hat{s}^{\ell+1}\|\leq B_{ps}\sum_{j=0}^{\ell}\eta^{\ell-j}\|\tilde{\boldsymbol{z}}^{j+1}-\tilde{\boldsymbol{z}}^j\|
\end{equation}
For the second term, by contractiveness, we have
\begin{equation}\label{off-on_2}
    \|\hat{s}^{\ell}-\mathrm{J}\phi_2(x)\|\leq \eta^{\ell}\|\hat{s}^0-\mathrm{J}\phi_2(x)\|
\end{equation}
Recall that $\phi_2(x)$ is locally Lipschitz continuous on $\Omega_{x,i}$ and
$\|\mathrm{J}\phi_2(x)\|\leq L_S$, Thus, there exists a constant $B_{ns}$ such that $\|\hat{s}^0-\mathrm{J}\phi_2(x)\|\leq B_{ns}$.
Summarizing \eqref{off-on_1} and \eqref{off-on_2}, we obtain
\begin{equation}
    \Vert \tilde{\boldsymbol{s}}^{\ell} - \mathrm{J}\phi_2(x) \Vert \leq B_{\text{ps}} \sum_{j=0}^{\ell-1} \eta^{\ell-1-j} \Vert \tilde{\boldsymbol{z}}^{j+1} - \tilde{\boldsymbol{z}}^{j} \Vert + B_{\text{ns}} \eta^{\ell}.
\end{equation}
Note that 
    $\|\tilde{z}^{\ell+1}-\tilde{z}^\ell\|\leq \eta^\ell\|\tilde{z}^1-\tilde{z}^0\|$
 implies that
\begin{equation}
\begin{aligned}
\Vert \tilde{s}^{\ell} - \mathrm{J}\phi_2(x) \Vert \leq 
B_{\text{ps}} \Vert \tilde{\boldsymbol{z}}^{1} - \tilde{\boldsymbol{z}}^{0} \Vert 
\sum_{j=0}^{\ell-1} \eta^{\ell-1-j} \eta^{j} + B_{\text{ns}}\eta^{\ell} \\
= B_{\text{ps}} \Vert \tilde{\boldsymbol{z}}^{1} - \tilde{\boldsymbol{z}}^{0} \Vert 
\eta^{\ell-1}\ell + B_{\text{ns}}\eta^{\ell}.
\end{aligned}
\end{equation}
Since \(0<\eta<1\), as \(\ell\to\infty\), we have \(\|\tilde{s}^{\ell}-\mathrm{J}\phi_2(x)\|\to 0\). \hfill$\square$

\noindent\textbf{Proof of Lemma~\ref{definable}}

Since locally Lipschitz definable mappings are path differentiable \cite{bolte2021conservative}, the Clarke Jacobian of a Lipschitz definable mapping is a conservative Jacobian. 
Recalling that $\phi_2(x)=\text{BR}_{\boldsymbol{Z}}(x,\text{BR}_Y(x),\theta)$ is determined by
$\boldsymbol{z}-h(x,\phi_1(x),\boldsymbol{z})=0$,
we define
\[
G(x,\boldsymbol{z}) := \boldsymbol{z}-h(x,\phi_1(x),\boldsymbol{z}).
\]

Since both $F$ and $P_{\Omega_{\boldsymbol{z}}}$ are locally Lipschitz and definable, the mapping $G(x,\boldsymbol{z})$ is also locally Lipschitz and definable. Moreover, because
$I-\mathcal{J}_{3}h(x,\phi_1(x),\phi_2(x))$
is invertible, it follows from the Lipschitz definable implicit function theorem \cite{bolte2021nonsmooth} that $\phi_2(x)$ is definable on $\Omega_{x,i}$. 
Since   definability is preserved by function composition,  $\hat{U}_X(x)=U_X(x,\phi_1(x),\phi_2(x))$ is definable. \hfill$\square$

\noindent\textbf{Proof of Lemma~\ref{sumable}}


Since $\hat{U}_X$ is differentiable at $x^k$,
\begin{equation}
\begin{aligned}
&\|e^{k}\|
\le \|\nabla \hat{U}_X^{k}-\xi^{k}\|\\
&=\|\nabla_1 U_X(x^k,y^{k+1},\boldsymbol{z}^{k+1})-
\nabla_1 U_X(x^k,y^{k+1},\phi_2(x))\\
&+(s^{k+1})^\top \nabla_3 U_X(x^k,y^{k+1},\boldsymbol{z}^{k+1})\\&-(\mathrm{J}\phi_2(x^k))^\top \nabla_3 U_X(x^k,y^{k+1},\phi_2(x))\|\\
&\le B_{ey}\|\boldsymbol{z}^{k}-\phi_{2}(x^{k})\|
+ B_{es}\|s^{k}-\mathrm{J}\phi_{2}(x^{k})\|,
\end{aligned}
\end{equation}
where $B_{ey}, B_{es}$ follow from the Lipschitz continuity of $\nabla{U}_X$ and  $\phi_2(x)$.
    Based on Algorithm 2's termination condition, at step $k$, 
\begin{equation}
    \max \{(\eta)^\ell,\ \sum_{j=0}^\ell \eta^{\ell-j}\|\tilde{\boldsymbol{z}}^{j+1}(k) -\tilde{\boldsymbol{z}}^{j}(k) \| \}\leq \sigma^k
\end{equation}
Since $h$ is a contraction mapping with the constant $\eta$,
$    \|\tilde{\boldsymbol{z}}^\ell-\phi_2(x)\|\leq \frac{1}{1-\eta}\|\tilde{\boldsymbol{z}}^{\ell+1}-\tilde{\boldsymbol{z}}^{\ell}\|$.
Thus, 
\begin{align*}
    \|e^k\|&\leq B_{ey}\|\boldsymbol{z}^k-\phi_2(x^k)\|
    +B_{es}\|s^k-\mathrm{J}\phi_2(x^k)\|\\
    &\leq \frac{B_{ey}}{1-\eta}\sigma^k+B_{es}B_{ps}\sigma^k+B_{es}B_{ns}\sigma^k
\end{align*}
Because $\sum \alpha^k\sigma^k< \infty$, $\sum \alpha^k\|e^k\|<\infty $.  \hfill$\square$

\noindent
\textbf{Proof of Theorem~\ref{convergence}}

We let $G(x) := -\mathcal{J}\hat{U}_X(x)-N_{{X}}(x)$ and observe that $0 \in G(x)$ implies that $x$ is a composite critical point. To prove convergence, we will invoke \cite{davis2018stochasticsubgradientmethodconverges}.
\begin{enumerate}
\item All limit points of $\{x^k\}_{k \in \mathbb{N}}$ lie in $\mathcal{X}$.
\item The iterates are bounded, i.e., $\sup_{k \in \mathbb{N}}\|x^k\|<\infty$ and $\sup_{k \in \mathbb{N}}\|\xi^k\|<\infty$.
\item The sequence $\{\alpha^k\}_{k \in \mathbb{N}}$ is nonnegative, nonsummable, and square-summable.
\item The weighted noise sequence is convergent: $\sum_{k=0}^{\infty} \alpha^k e^k \to v$ for some $v \in \mathbb{R}^m$.
\item For any unbounded increasing sequence $\{k_j\} \subseteq \mathbb{N}$ such that $x^{k_j}$ converges to some point $\bar{x}$,
$$
\lim_{n \to \infty} \text{dist}\left(\frac{1}{n} \sum_{j=1}^{n} \xi_{k_j}; G(\bar{x})\right)=0.
$$
\item There exists a continuous function $\phi: \mathbb{R}^m \to \mathbb{R}$, which is bounded from below, and such that 

 , for a dense set of values $r \in \mathbb{R}$, the intersection $\phi^{-1}(r) \cap G^{-1}(0)$ is empty , and moreover, when $z: \mathbb{R}_{\geq 0} \to \mathbb{R}^m$ is a trajectory of the differential inclusion $\dot{z}(t) \in G(z(t))$ and $0 \notin G(z(0))$, there exists a $T > 0$ satisfying
$$
\phi(z(T)) < \sup_{t \in [0,T]} \phi(z(t)) \leq \phi(z(0)).
$$
\end{enumerate}
Condition 1 is obviously met. For condition 2, we have
\begin{equation}
    \|\xi^k\|=\frac{1}{\alpha^k}\|\mathbb{P}_{\Omega_x}(x^k)-\mathbb{P}_{\Omega_x}(x^k-\alpha^k\zeta^k)\|
    \leq \|\zeta^k\|
\end{equation}
where $\zeta^k\in\mathcal{J}\hat{U}_x$,   thus,
$\sup \|\xi^k\|< \infty$. Further,
Condition 3 holds by design, while Condition 4 is shown
in Lemma~\ref{sumable}.

To show that Condition 5 is satisﬁed, we ﬁrst note that $\mathcal{J}\phi_2(x)$ is convex  \cite{bolte2021conservative}, 
hence $G(\cdot)$ is convex-valued. Therefore,
\begin{equation}
    \text{dist}\left(\frac{1}{n} \sum_{j=1}^{n} \xi_{k_j}, G(\bar{x})\right) \leq \frac{1}{n} \sum_{j=1}^{n} \text{dist}\left(\xi_{k_j}, G(\bar{x})\right).
\end{equation}
Define $w^{k_j}=\mathbb{P}_{\Omega_{x}}[x^{k_j}+\alpha^k\zeta^{k_j}]$, and then 
\begin{equation}
w^{k_j}=\underset{w}{\mathrm{argmin}}\ 
    \mathcal{I}_{\Omega_x}(w)+\|w-\alpha^{k_j}\zeta^{k_j}\|^2
\end{equation}
According to Fermat's rule and  $\partial \mathcal{I}_{\Omega_x} = N_{X}$,
\begin{equation}
\begin{aligned}
0 \in N_X(w^{k_j}) + \frac{1}{\alpha^{k_j}}(w^{k_j} - x^{k_j} + \alpha^{k_j}\zeta^{k_j}) \Rightarrow \\
-\frac{1}{\alpha^{k_j}}(x^{k_j} - w^{k_j}) \in -N_X(w^{k_j}) - \zeta^{k_j}
\end{aligned}
\end{equation}
Since $\Omega_{x,i}$ is compact and $\mathcal{J}\hat{U}_X$, ${N}_{X}$ are outer
continuous,
\begin{equation}
    \xi^{k_j} \to N_{X}(\bar{x}) + \bar{\xi} \in G(\bar{x})
\end{equation}
Consequently,  Condition 5 follows. 

Then we utilize $\hat{U}_X(x)$ as the Lyapunov function $\phi$ and recall that $\hat{U}_X$ is definable by virtue of Lemma~\ref{definable}. Thus, $\hat{U}_X$ admits a Whitney $C^1$ stratification.
Thus, Condition 6 follows exactly from the arguments
in the proof of \cite{davis2018stochasticsubgradientmethodconverges}.

By Assumption~\ref{concave}, Algorithm 1 converges to the maximum of $\hat{U}_X$ on $\Omega_{x,i}$ for a fixed $\theta$. Consequently, by executing the algorithm across all intervals and comparing the resulting utilities with those at the roots of $f_2(x, \theta)$, we determine the leader's optimal strategy for a given $\theta$. Finally, iterating over $\theta$ yields the optimal deception parameter $\theta^*$ and the corresponding strategy $x^*$.

For the WDSE, we adopt the inf form of \eqref{eq:zero_utility}. The convergence analysis for the SDSE is similar to that for the WDSE. We only need to replace the inf form of \eqref{eq:zero_utility} with the sup form in the proof. \hfill$\square$


\section{} \label{epsilonsection}
\noindent\textbf{Proof of Theorem~\ref{theorem4.2}}

For a fixed parameter $\theta$,  let the zeros of $f_2(x,\theta)$ be  $\{x_1,x_2,\ldots,x_{q_\theta}\}$. In the absence of analytical solutions for these zeros, we can instead determine closed intervals $I_j$ such that $x_j\in I_j$.
Then \( \Omega_x = \left( \bigcup_i \tilde{\Omega}_{x,i} \right) \bigcup \left( \bigcup_j I_j \right) \), where \( \tilde{\Omega}_{x,i} \) denotes the closed interval formed by the right endpoint of \( I_i \) and the left endpoint of \( I_{i+1} \).
Define 

\[
{U}^{\dag} = \max \left( \left\{ \max_{x \in \tilde{\Omega}_{x,i}} \hat{U}_x(x) \right\}_{ i}, \left\{  g(x) | x\in I_j \right\}_{ j} \right),
\]
where \( \hat{U}_X(x) \) is the leader’s utility under a fixed sign of \( f_2(x,\theta) \), and \( g(x) = \inf\limits_y U_X(x, y, BR_{\boldsymbol{Z}}(x,y,\theta)) \) represents the leader’s conservative utility over the uncertain regions  \( I_j \).

Let $ U^*=\sup\limits_{x\in\Omega_x}\inf\limits_{y\in\Omega_y} U_X(x,y,\text{BR}_{\boldsymbol{Z}}(x,y,\theta)) $ denote the leader's optimal utility.  
Recall that  $\text{BR}_{\boldsymbol{Z}}(x, y)$ is a fixed point of the mapping $h(z) = \mathbb{P}_{\Omega_{\boldsymbol{z}}}(z - \gamma F(x, y, z))$. Let $z_1 = \text{BR}_{\boldsymbol{Z}}(x_1, y)$ and $z_2 = \text{BR}_{\boldsymbol{Z}}(x_2, y)$. Using the non-expansiveness of the projection operator $\mathbb{P}_{\Omega_{\boldsymbol{z}}}$, 
\begin{equation}
\begin{aligned}
    &\|z_1 - z_2\| \\
    &= \|\mathbb{P}_{\Omega_{\boldsymbol{z}}}(z_1 - \gamma F(x_1, y, z_1)) - \mathbb{P}_{\Omega_{\boldsymbol{z}}}(z_2 - \gamma F(x_2, y, z_2))\| \\
    &\leq \|(z_1 - \gamma F(x_1, y, z_1)) - (z_2 - \gamma F(x_2, y, z_2))\| \\
    &\leq \|z_1 - z_2 - \gamma(F(x_1, y, z_1) - F(x_1, y, z_2))\| \\
    &\quad + \gamma \|F(x_1, y, z_2) - F(x_2, y, z_2)\|.
\end{aligned}
\end{equation}
Due to the contraction  of the gradient descent step with respect to $z$  and the Lipschitz continuity of $F$ with respect to $x$, we obtain
\begin{equation}
    \|z_1 - z_2\| \leq \eta \|z_1 - z_2\| + \gamma L_x \|x_1 - x_2\|.
\end{equation}
Rearranging the terms yields
\begin{equation} \label{eq:BR_z_Lip}
    \|\text{BR}_{\boldsymbol{Z}}(x_1, y) - \text{BR}_{\boldsymbol{Z}}(x_2, y)\| \leq \frac{L_x \gamma}{1 - \eta} \|x_1 - x_2\|.
\end{equation}
Thus, $\text{BR}_{\boldsymbol{Z}}$ is Lipschitz continuous with respect to $x$.

Define the intermediate utility function $\beta(x, y) = U_X(x, y, \text{BR}_{\boldsymbol{Z}}(x, y, \theta))$. Clearly, 
\begin{equation}
\begin{aligned}
    &\|\beta(x_1, y) - \beta(x_2, y)\| \\
    &= \|U_X(x_1, y, \text{BR}_{\boldsymbol{Z}}(x_1, y)) - U_X(x_2, y, \text{BR}_{\boldsymbol{Z}}(x_2, y))\| \\
    &\leq \|U_X(x_1, y, \text{BR}_{\boldsymbol{Z}}(x_1, y)) - U_X(x_2, y, \text{BR}_{\boldsymbol{Z}}(x_1, y))\| \\
    &\quad + \|U_X(x_2, y, \text{BR}_{\boldsymbol{Z}}(x_1, y)) - U_X(x_2, y, \text{BR}_{\boldsymbol{Z}}(x_2, y))\|.
\end{aligned}
\end{equation}
Since $U_X$ is continuously differentiable on the compact set, let $L_x$ and $L_z$ be its Lipschitz constants with respect to $x$ and $z$, respectively. Substituting \eqref{eq:BR_z_Lip} yields
\begin{equation}
    \|\beta(x_1, y) - \beta(x_2, y)\| \leq \left(L_x + L_z \frac{L_x \gamma}{1 - \eta}\right) \|x_1 - x_2\|.
\end{equation}
Let $L_2 = L_x + L_z \frac{L_x \gamma}{1 - \eta}$. Then $\beta(x, y)$ is $L_2$-Lipschitz in $x$.

Now, consider the worst-case utility $g(x) = \inf\limits_{y \in \Omega_y} \beta(x, y)$.
\begin{equation}
\begin{aligned}
    g(x_1) - g(x_2) &= \inf_{y \in \Omega_y} \beta(x_1, y) - \inf_{y \in \Omega_y} \beta(x_2, y) \\
    &\leq \beta(x_1, y^*(x_2)) - \beta(x_2, y^*(x_2)) \\
    &\leq L_2 \|x_1 - x_2\|.
\end{aligned}
\end{equation}
By swapping $x_1$ and $x_2$, $|g(x_1) - g(x_2)| \leq L_2 \|x_1 - x_2\|$. Note $\hat{U}_X(x)$ on intervals $\Omega_{x,i}$ is Lipschitz continuous with a constant $L_3$. Let $L = \max(L_2, L_3)$ be the global Lipschitz constant.

The numerical procedure approximates the true optimum $U^*$ with $U^\dagger$ by sampling over the certainty interval $\tilde{\Omega}_{x,i}$ and uncertainty intervals $I_j$.

For certainty intervals $\Omega_{x,i}$, let $\mathcal{M}_i = \max_{x \in \Omega_{x,i}} \hat{U}_X(x)$ and $\tilde{\mathcal{M}}_i = \max_{x \in \tilde{\Omega}_{x,i}} \hat{U}_X(x)$. Since any point in $\Omega_{x,i}$ is at distance at most $\delta$ from $\tilde{\Omega}_{x,i}$, 
\begin{equation}
    |\mathcal{M}_i - \tilde{\mathcal{M}}_i| \leq L \delta.
\end{equation}

For uncertainty intervals $I_j$, we compare the true minimal utility $g_j = g(x_j)$  with the numerical surrogate $\tilde{g}_j = \inf_{x \in I_j} g(x)$. Since $x_j \in I_j$ and $\lambda(I_j)<\delta$, 
\begin{equation}
    g_j - L\delta \leq g(x) \quad \forall x \in I_j \implies g_j - L\delta \leq \tilde{g}_j.
\end{equation}
Since $\tilde{g}_j \leq g_j$ by definition, $|g_j - \tilde{g}_j| \leq L \delta$.

Let $U^* = \max \{ \mathcal{M}_i, g_j \}$ and $U^\dagger = \max \{ \tilde{\mathcal{M}}_i, \tilde{g}_j \}$. Since every element in the numerical set is within $L\delta$ of its true counterpart, 
\begin{equation}\label{lemma 6}
    |U^* - U^\dagger| \leq L \delta.
\end{equation}

Let $(x^*, \theta^*)$ be the true WDSE strategy, and  $(x_{\theta}, \theta)$ be the strategy computed by our algorithm for a fixed parameter $\theta$.
From \eqref{lemma 6}, the error between the computed utility and the true utility is bounded by $L\delta$. By choosing the partition mesh size such that $L\delta \leq \epsilon$, 
\begin{equation}
    \hat{U}_X(x_{\theta}, \theta) \geq \sup_{x \in \Omega_x} \inf_{y \in \text{BR}_Y(x, \theta)} U_X(x, y, \text{BR}_{\boldsymbol{Z}}(x, y, \theta)) - \epsilon.
\end{equation}
Since our algorithm selects $(\theta^*_{\text{alg}}, x^*_{\text{alg}})$ to maximize this computed utility over $\Theta$, and the true optimum $(x^*, \theta^*)$ is a feasible candidate in this search, it holds that
\begin{equation}
    \hat{U}_X(x^*_{\text{alg}}, \theta^*_{\text{alg}}) \geq \hat{U}_X(x^*, \theta^*) - \epsilon.
\end{equation}
Thus, the computed strategy is an $\epsilon$-WDSE. \hfill$\square$

\ifCLASSOPTIONcaptionsoff
  \newpage
\fi

\bibliographystyle{IEEEtran}
\bibliography{myrefs}

\end{document}